\documentclass[11pt]{amsart}
\usepackage{amssymb}  
\usepackage{amsmath}
\usepackage[normalem]{ulem}
\usepackage{enumitem}
\usepackage{hyperref}
\usepackage{a4wide}
\usepackage{bm}
\usepackage{amsthm}
\usepackage{tikz}

\newtheorem{theorem}{Theorem}[section]
\newtheorem{lemma}[theorem]{Lemma}
\newtheorem{corollary}[theorem]{Corollary}
\newtheorem{proposition}[theorem]{Proposition}

\theoremstyle{definition}
\newtheorem{definition}[theorem]{Definition}
\newtheorem{remark}[theorem]{Remark}
\newtheorem{example}[theorem]{Example}

\newcommand{\SN}{\mathcal{S}_N}
\newcommand{\DN}{\mathcal{D}_N}

\newcommand{\cA}{\mathcal{A}}

\newcommand{\cM}{\mathcal{M}}

\newcommand{\Pois}{\mathit{Poisson}}
\newcommand{\Prob}{\mathrm{P}}
\newcommand{\mypath}{\mathcal{P}_k}

\newcommand{\tr}{\mathrm{tr}}
\newcommand{\reg}{\mathsf{reg}}

\newcommand{\RR}{\mathbb{R}}

\newcommand{\e}{\mathrm{e}}
\newcommand{\M}{\mathcal{M}}

\newcommand\CC{\mathbb{C}}
\newcommand\NN{\mathbb{N}}

\newcommand\ds{\displaystyle}

\newcommand\ess{\mathbf{s}}
\newcommand\Aess{\widetilde{\mathbf{s}}}

\newcommand\zee{\mathbf{z}}
\newcommand{\Rad}{\mathrm{Rad}}

\newcommand{\bt}[1]{{\color{black}{#1}}}

\newcommand{\exend}{\hfill $\Diamond$}

\begin{document}

\title[A new algebraic approach]{A new algebraic approach to genome rearrangement models}

\author{Venta Terauds and Jeremy Sumner}

\address{Discipline of Mathematics, School of Natural Sciences, Private Bag 37: University of Tasmania, Sandy Bay, Tasmania 7001, Australia}

\email{venta.terauds@utas.edu.au\\ jeremy.sumner@utas.edu.au  }

\thanks{This work was supported by Australian Research Council Discovery Grant DP180102215. The authors would like to thank Andrew Francis and Joshua Stevenson, for many interesting discussions relating to this work, and the anonymous reviewers, for their constructive comments.}

\begin{abstract}
We present a unified framework for modelling genomes and their rearrangements in a genome algebra, as elements that simultaneously incorporate all physical symmetries. Building on previous work utilising the group algebra of the symmetric group, we explicitly construct the genome algebra for the case of unsigned circular genomes with dihedral symmetry and show that the maximum likelihood estimate (MLE) of genome rearrangement distance can be validly and more efficiently performed in this setting. We then construct the genome algebra for \bt{a more} general case, that is, for genomes \bt{that may be} represented by elements of an arbitrary group and symmetry group, and show that the MLE computations can be performed entirely within this framework. There is no prescribed model in this framework; that is, it allows any choice of rearrangements \bt{that preserve the set of regions}, along with arbitrary weights. Further, since the likelihood function is built from path probabilities -- a generalisation of path counts -- the framework may be utilised for any distance measure that is based on path probabilities. 
\end{abstract}

\maketitle

\section{Introduction}\label{sec:intro}

In the eight decades since Dobzhansky and Sturtevant observed that differences in fruit fly genomes could be explained by a sequence of reversals of genome segments  \cite{Dobzhansky38}, the study of evolution via genome rearrangement has developed into a rich and active field, with diverse applications \cite{chen-mitoch,Darmon-bacteria,Oesper-cancer}. Much work focuses on the calculation of evolutionary distances under rearrangement models, with the distances subsequently used to reconstruct phylogenetic trees. 
For example, minimal rearrangement distances between genomes -- and other distance estimates based on these -- have been studied extensively and, under various model restrictions, can be calculated efficiently \cite{bader01,Wang_Warnow_06,bader_ohlebusch_07,Oliveira_etal_19}. There are, however, good arguments for applying stochastic methods that estimate genomic distance, via rearrangement, as evolutionary time elapsed \cite{mles}, particularly when such an approach allows various rearrangement models to be considered \cite{vjez_circ1}. 

The maximum likelihood approach detailed in \cite{mles} utilised the theory of the symmetric and dihedral groups to model circular genomes and \bt{region set-conserving} rearrangements, motivated by earlier group-theoretical approaches to rearrangement models \cite{andrew14,attilaand}. The combinatorial problem of calculating the maximum likelihood estimate (MLE) of evolutionary distance was then converted into a numerical one in \cite{jezandpet} via the representation theory of the symmetric group algebra. In \cite{vjez_circ1}, the consideration of symmetry was extended to include symmetry of rearrangement models, the role of this in simplifying calculations was explored, and the concrete implementation of the technique for a general model was described. Whilst the representation theory approach reduces the complexity of the MLE computations, the complexity is still of factorial order, meaning that computations for large number of regions remain, for the moment, out of reach. 

In this work, we suggest that the appropriate theoretical setting for such MLE computations is in fact not the symmetric group algebra but a lower-dimensional algebra. In the symmetric group algebra, the basis elements for computations are individual permutations, each representing a rearrangement or a genome in a fixed orientation, and symmetry is incorporated as an extra step in the calculations. To simplify this, we construct an algebra that incorporates the inherent symmetry into each element. Here, the basis elements are {\em permutation clouds}. These correspond to genomes, by simultaneously including all physical orientations; due to the corresponding symmetries of the rearrangement model, they also represent rearrangements in a natural way. 

This approach explains and removes the redundancy in the MLE computations that was observed in \cite{vjez_circ1}. In developing the approach, we firstly focus on the \bt{simple} concrete case of \bt{uni-chromosomal circular} genomes \bt{modelled} with unoriented regions and no distinguished positions, building on previous work \cite{mles,jezandpet,vjez_circ1}. Subsequently, we demonstrate that our results may be applied more generally, for example \bt{to genomic models} that include region orientation and/or origin and terminus of replication. Further, although our focus is on calculation of MLEs, our approach can be applied to calculate other measures of genomic distance under rearrangement; in particular, any that utilise path counts or weighted path counts, such as minimum distance. \bt{Whilst the framework does not specify a rearrangement model --- indeed, one may choose the allowed rearrangements and their weights --- we note that the group-based approach limits us to rearrangements that conserve the set of genomic regions, and thus cannot accommodate insertions, deletions or duplications. We are currently working on expanding the framework to a semigroup-based approach that could incorporate at least some of these rearrangement types. Some of the algebra easily extends to the semigroup case (see Remark~\ref{rem:semigroups}, for example), however there is much yet to be done and this is outside the scope of the present paper.}

In the next section, we outline the details of the symmetric group algebra approach to calculating MLEs \cite{jezandpet,vjez_circ1} for pairs of unsigned, circular genomes, which forms the foundation for the current work. Following this, in Section~\ref{sec:A}, we construct the genome algebra, based on permutation clouds, \bt{for this case} and show that it provides a coherent framework for modelling \bt{such genomes and region set-conserving} rearrangements and for calculating MLEs. Section~\ref{sec:gen} outlines the extension of our results and techniques from permutations with dihedral symmetry to an arbitrary group and symmetry group. This verifies that, as well as incorporating flexibility in the rearrangement model, the framework is not specific to \bt{one} particular genomic model. The paper concludes with a brief discussion section.

\section{Background: the permutation approach}\label{sec:perms}

In this section we set out the theoretical framework for rearrangement models based on permutations, and recall the key elements of the technique for calculating the maximum likelihood estimate of evolutionary distance. Full derivations and details may be found in \cite{vjez_circ1} and the earlier papers \cite{jezandpet,mles}. 
For the specific case study in this and the next section, we model the evolution of single-strand, circular genomes; \bt{we do not consider the regions to be oriented and do not distinguish any positions.}\footnote{\bt{These latter two may be considered simplifying assumptions, since they reduce the size of the state space. However, as will be shown later, the framework just as easily accommodates the alternate cases.}} Genomes that are to be compared share $N$ identified regions\footnote{Here, a {\em region} is a contiguous section of the genome such as a sequence of genes (for an example of how such genomic simplification may be enacted in practice, see \cite{belda05})} of interest \bt{and we consider only rearrangements that conserve the set of regions}. Accordingly, we use unsigned permutations, that is, elements of the symmetric group, $\SN$, to represent both genomes and rearrangements. Explicitly: the regions and positions are each labelled by the integers $\{1,2,\ldots, N\}$, and a given genome is represented by a permutation $\sigma\in\SN$, where
\[ \sigma(i)=j \; \iff \; \textrm{ region } i \textrm{  is in position } j\,.\]
Note that while the region labels are chosen once and are immutable, the position labelling reflects a choice of reference frame (starting position and direction of numbering) that changes when we move the genome in space. 
Since we do not distinguish any positions, there are $2N$ possible choices of reference frame and thus $2N$ distinct permutations that represent any given genome; we denote these by
\begin{equation}\label{eq:genomeclass}
[\sigma]:=\{d\sigma : d\in \DN\}\,.
\end{equation}
Here, $\DN$ is the dihedral group, and the genome has {\em dihedral symmetry}. 

Since $\DN$ is a subgroup of $\SN$, the sets $[\sigma]$ are cosets, that is, each $[\sigma] = \DN\sigma$ is an equivalence class of $\SN$. 
Since a given genome exists independently of its orientation in space, we may identify it with the entire coset \cite{mles,attilaand}. However, in this initial formulation, we choose any one of the permutations from the coset (\ref{eq:genomeclass}), say $\sigma$, to specify the genome, work with this single permutation at first, and incorporate all permutations in the set $[\sigma]$ (all symmetries of the genome) into the likelihood calculations in due course.

We model evolution as a sequence of discrete rearrangement events occurring in continuous time. In this section, as in previous work, we consider a rearrangement to be a single permutation acting on a single permutation; in the next section we shall develop this into the notion of permutation clouds acting on permutation clouds. For now, however, for a genome represented by $\sigma\in\SN$, a rearrangement event is represented by a permutation, $a \in \SN$, acting on $\sigma$ (on the left):
$\sigma \mapsto a\, \sigma$.
We refer to permutations $a$ acting in this way as ``rearrangements''. 

The full biological model for evolution is given by $(\cM,w,dist)$, where $\M \subseteq \SN$ is the set of allowed rearrangements, $w:\cM\to (0,1]$ is the probability distribution on this set, and $dist$ is the probability distribution of the independent rearrangement events in time. 
\bt{One may have biological evidence for including particular types and sizes of rearrangements in the model with differing relative probabilities, or may wish to compare distances computed under differing models (see \cite{vjoshjez} for some specific examples of models and distance comparisons). The distribution $dist$ may similarly be chosen according to evidence or preference;} in this treatment, we use the Poisson distribution.

We shall emphasise at this stage that \bt{we make minimal further restrictions} on the set of allowed rearrangements $\cM$. Without loss of generality, we assume that $\cM$ generates the group $\SN$; this means that any permutation in $\SN$ may be obtained from any other by applying a sequence of elements from $\cM$. \bt{(Note that the case of $\cM$ not generating $\SN$ is simpler: in this case $\cM$ generates a subgroup, $H\subseteq\SN$, and the problem reduces to considering this smaller group, since any pair of elements would simply be unrelated under the model or both be elements of a coset $H\sigma$.)}  Elements of $\DN$ are, formally, allowed in the set of rearrangements, although their action does not actually alter the genome. This allows for full generality; for example, if one wishes to include `all inversions' in the model, then the inversion of a region of size $N-1$ is the same as flipping the genome over in space.

The first model condition simply states that the model should naturally possess the same symmetry as the genome (in the current case, dihedral symmetry).  Suppose, for example, that $(1,2)\in\cM$, meaning that the regions in positions $1$ and $2$ may swap places. Then, since the position labelling is arbitrary, we should have $(\ell,\ell+1)\in\cM$ for all $\ell = 1,2,\ldots,N-1$, and $(N,1)\in\cM$, meaning that  {\em any} two regions in adjacent positions may swap places; further, these rearrangements should all be equally probable. We refer to this property as {\em dihedral symmetry of the model} and, mathematically, express the condition as 
\begin{equation} 
\textrm{ for each } a\!\in\!\M \textrm{ and } d\!\in\! \DN , 
			\;dad^{-1} \!\in\!\M \textrm{ and } w(dad^{-1}) = w(a)\,.\tag{M1}\label{cond:m1}
\end{equation}

The second model condition ensures that the modelling is agnostic to the temporal direction of evolution. More precisely, the condition states that for any rearrangement that is allowed, its inverse is also allowed, with the same probability. That is, 
\begin{equation} 
\textrm{ for each } a\!\in\!\M , \;a^{-1} \!\in\!\M \textrm{ and } w(a^{-1}) = w(a)\,.\tag{M2}\label{cond:m2}
\end{equation}
We refer to this property as {\em rearrangement reversibility of the model}, or simply model reversibility. \bt{This condition is natural in the current group-based setting, where the typical rearrangements are reversals (which are self-inverse) and translocations (whose inverses are translocations). It is not essential for most of the construction, however does have some nice implications. For example,} when we interpret our model of evolution as a Markov process, in Section~\ref{subsec:markov}, we will show that (\ref{cond:m2}) is equivalent to the time reversibility of the Markov process. 

The evolutionary distance measure we consider in this paper is the maximum likelihood estimate of time elapsed (MLE). 
This is the maximum value of the likelihood function, which gives the probability, for any given time $T$, that the reference genome has evolved into the target genome in this amount of time. 
To be precise, for reference genome represented by the identity permutation $e\in\SN$ and target genome represented by $\sigma\in\SN$, the MLE is the maximum of the function $L(T|\sigma)$, where
\begin{align}
 L(T|\sigma) &:= P(\sigma|T)
     = \sum_{k=0}^{\infty} \Prob(e\!\mapsto\![\sigma] \textrm{ via } k \textrm{ events})
     														\Prob(k \textrm{ events in time } T)\,.\label{eq:simplelike}
\end{align}
Of course, the likelihood function need not have a maximum; this simply means that no evidence of an evolutionary relationship between the reference and the target under the given model can be discerned. This scenario, familiar from DNA sequence alignment paradigms such as the Jukes-Cantor correction \cite{felsenstein2004inferring}, was discussed in the context of genome rearrangement models in \cite{mles} and \cite{vjez_circ1}, where it was observed to occur in a substantial proportion of cases (independently of the chosen biological model).

For each $k$, the factor $\Prob(k \textrm{ events in time } T)$ in the likelihood expression is determined by the distribution $dist$. The first factor, 
$\Prob(e\!\mapsto\![\sigma] \textrm{ via } k \textrm{ events})$ is the {\em genome path probability}, which we shall denote by $\alpha_k(\sigma)$.
Since the target genome may be represented by any permutation from the set $[\sigma]$, `$e\mapsto[\sigma]$' is shorthand for ``the permutation $e$ is transformed into any permutation from the set $[\sigma]$'' and we calculate the genome path probability as a sum of {\em permutation path probabilities}, denoted by $\beta_k(\sigma)$. That is,
\[ \alpha_k(\sigma) := \Prob(e\!\mapsto\![\sigma] \textrm{ via } k \textrm{ events}) 
		= \sum_{d\in\DN} \Prob(e\!\mapsto\! d\sigma \textrm{ via } k \textrm{ events}) = \sum_{d\in\DN}\beta_k(d\sigma)\,.\]
Given a permutation $\sigma\in\SN$ and a model $(\cM,w,dist)$, we specify each permutation path probability $\beta_k(\sigma)$ by considering the set $\mypath(\sigma)$ of all $k$-length sequences of permutations, chosen from $\cM$, that transform $e$ into $\sigma$. 
Since we assume rearrangement events to be independent, the permutation path probability is then the sum of the probabilities of all such sequences, that is, 
\[ \beta_k(\sigma) = \!\!\sum_{(a_1,a_2,\ldots,a_k)\in\mypath(\sigma)} \!\!w(a_1)w(a_2)\ldots w(a_k)\,.\]

We note that permutation path probabilities vary for different elements of $[\sigma]$ (that is, in general, 
$\beta_k(\sigma)\neq \beta_k(d\sigma)$ for $\sigma\in\SN, d\in\DN$). However, {\em genome} path probabilities are of course constant on the cosets $[\sigma]$.
In fact, the model symmetry conditions (\ref{cond:m1}) and (\ref{cond:m2}) ensure that there are bigger classes of permutations that all have the same path probabilities (and thus likelihoods). The following results were established in \cite{mles} ($(i)$) and \cite{vjez_circ1} ($(ii)$ and $(iii)$).

\begin{theorem}\label{thm:equivs}
Let $(\cM,w,dist)$ be a full biological model for evolution. For all $k\in\NN_0$ and $\sigma\in\SN$, the following hold. 
\begin{enumerate}
\item[(i)]  $\alpha_k(\sigma_1)=\alpha_k(\sigma_2)$ for all $\sigma_1,\sigma_2\in [\sigma] = \{d\sigma : d\in\DN\}$.
\item[(ii)] If the model has the dihedral symmetry property (\ref{cond:m1}), then $\alpha_k(\sigma_1)=\alpha_k(\sigma_2)$ for all $\sigma_1,\sigma_2$ in the set
\[ [\sigma]_D := \{d_1\sigma d_2 : d_1,d_2\in\DN\}\,.\]
\item[(iii)]If the model has the dihedral symmetry property (\ref{cond:m1}) and the reversibility property (\ref{cond:m2}), then $\alpha_k(\sigma_1)=\alpha_k(\sigma_2)$ for all $\sigma_1,\sigma_2$ in the set
\[ [\sigma]_{DR} := \{ d_1 \sigma d_2\,, d_1\sigma^{-1}d_2 : d_1, d_2 \in \DN\}\,. \] 
\end{enumerate}
\end{theorem}

It was shown in \cite{jezandpet} that the combinatorial problem of calculating path probabilities may be converted into a linear algebra problem via the representation theory of the symmetric group algebra, $\CC[\SN]$. For full details in the general model setting, we refer the reader to \cite{vjez_circ1}. We recall the essential steps in the derivation here, since we shall undertake a similar procedure in a lower dimensional algebra in the next section. 

Here, we use the term {\em algebra} to mean a vector space equipped with a bilinear product. In particular, we require the group algebra $\CC[\SN]$ consisting of all formal linear combinations of elements of the group $\SN$; this algebra has natural basis $\SN$, and thus dimension $N!$. For detailed background on the symmetric group algebra, and algebras more generally, we refer the reader to \cite{sagan} and \cite{etingof_book} respectively.
The following group algebra elements are key to our calculations.

\begin{definition}\label{def:mod-sym-els}
Let $(\cM,w,dist)$ be a biological model for evolution of genomes with $N$ regions. We define the {\em model element}, $\ess$, and the {\em symmetry element}, $\zee$, of the group algebra $\CC[\SN]$ by
\[ \ess:=\ds\sum_{a\in\mathcal{M}} w(a)\,a\,  \quad\quad \mathrm{ and } \quad\quad \zee:= \tfrac{1}{2N}\sum_{d\in \DN} d \,.\]
\end{definition}

To reformulate the path probabilities, we firstly observe that 
\begin{equation}\label{eq:s^k} 
\ess^k=\ds\sum_{\tau\in\SN} \beta_k(\tau)\,\tau\,.
\end{equation}
Then, for $\sigma\in\SN$, we multiply (\ref{eq:s^k}) on the left by $\sigma^{-1}$ to see that $\beta_k(\sigma)$ is the coefficient of $e$ in the expansion of $\sigma^{-1}\ess^k$. The representation theory of the symmetric group algebra tells us that this is exactly ($\frac{1}{N!}$ times) the trace of the regular representation of $\sigma^{-1}\ess^k$. That is,
\begin{equation}\label{eq:beta-k} \beta_k(\sigma) = \tfrac{1}{N!}\chi_{\reg}(\sigma^{-1}\ess^k)\,.\end{equation}
Thus, for $\sigma\in\SN$, the $k$th genome path probability is 
\begin{align}\label{eq:alpha_k-SN}
\alpha_k(\sigma) = \sum_{d\in \DN}\!\! \beta_k(d\sigma) &=  \tfrac{1}{N!}\!\sum_{d\in \DN}\!\!\chi_{\reg}(\sigma^{-1}\!d\ess^k)
 =  \tfrac{2N}{N!} \chi_{\reg}(\sigma^{-1}\zee\ess^k) 
 = \tfrac{2N}{N!} \sum_{p\vdash N} \! D_p \chi_{p}(\sigma^{-1}\zee\ess^k)\,,\end{align}
where we have used the linearity of the characters to incorporate the symmetry element $\zee$. \bt{The final equality is gained by decomposing the regular representation of $\CC[\SN]$ into irreducible representations. Recall that the irreducible representations of $\CC[\SN]$ correspond to the integer partitions of $N$ \cite[Prop. 1.10.1]{sagan}; here, we denote a partition of $N$ by $p\vdash N$ and index the representations and related objects accordingly. Specifically, for each partition $p\vdash N$, $\rho_p$ is the irreducible representation corresponding to $p$,  $D_p$ is its multiplicity (and dimension), and $\chi_p$ is the character of this representation.} 

The above derivation of the permutation path probabilities follows that of \cite{jezandpet}; in that paper it was also noted that an alternative derivation is possible via the theory of the Fourier transform on $\SN$. That is, one may extend the probability distribution $w$ on $\cM$ to $w'$ on the whole of $\SN$, notice that the Fourier transform of $w'$ with respect to an irreducible representation $\rho_p$ is equal to $\rho_p(\ess)$ and that $w'$ convolved with itself $k$ times is exactly the function $\beta_k$ on $\SN$, and then apply the Fourier inversion formula to obtain (\ref{eq:beta-k}).

Now, for a model with rearrangement reversibility (\ref{cond:m2}), the irreducible representations of the model element $\ess$ are diagonalisable \cite{vjez_circ1} and we obtain
\[ \begin{split}
\alpha_k(\sigma)
&= \tfrac{2N}{N!}\sum_{p\vdash N} D_p\sum_{i=1}^{r_p} (\lambda_{p,i})^k \tr(\rho_p(\sigma^{-1}\zee) E_{p,i})\,,
\end{split} 
\]
where, for each $p$, the eigenvalues of $\rho_p(\ess)$ are $\{\lambda_{p,i} : i = 1,\ldots, r_p\}$ and, for each $p$ and $i$, $E_{p,i}$ is the projection onto the eigenspace of $\lambda_{p,i}$.
Substituting this into the likelihood expression (\ref{eq:simplelike}) and setting the distribution of events in time to be $dist = \Pois(1)$, we obtain
\begin{align}
L(T|\sigma) 
&= e^{-T}\, \tfrac{2N}{N!}\sum_{p\vdash N} D_p\sum_{i=1}^{r_p}\tr(\rho_p(\sigma^{-1}\zee) E_{p,i}) e^{\lambda_{p,i}T}\,,\label{eq:perms-like}
\end{align}
where we have observed that the expression is in fact a power series and, accordingly, have been able to eliminate the infinite sum from the expression. 

We note that, for a given model, one need only calculate the eigenvalues of each $\rho_p(\ess)$ once. Thus the bulk of the calculation burden is now in calculating the {\em partial traces}, that is, for any given genome, the set 
$\left\{\tr(\rho_p(\sigma^{-1}\zee) E_{p,i}): p\vdash N, i = 1\ldots r_p\right\}$
of coefficients that correspond to the distinct eigenvalues in the likelihood equation. 

In implementing the likelihood calculations using the expression (\ref{eq:perms-like}), we observed that for all genomes, {\em most} of these partial trace coefficients were zero \cite{vjez_circ1}. That is, most of our calculations ended up not contributing to the final likelihood function. In the next section, we explain the occurrence of these zeroes and show that the redundancy can be removed from the computations.

\section{The circular genome algebra}\label{sec:A}

The calculations outlined above are performed in the group algebra $\CC[\SN]$, where each permutation in $\SN$ is a distinct basis element. However, we (and, indeed, the computations) do not distinguish between different permutations that represent the same genome -- that is, between elements of each equivalence class
\[ [\sigma] = \{ d\sigma : d\in\DN\}\,,\]
for $\sigma\in\SN$. We now construct a lower-dimensional algebra by combining these equivalent permutations together to form basis elements -- {\em permutation clouds} -- that correspond to circular genomes.  Until otherwise stated, assume that we have fixed a number of regions $N$ and a biological model for evolution $(\cM,w,dist)$ that has dihedral symmetry and is reversible, that is, satisfies (\ref{cond:m1}) and (\ref{cond:m2}). 

\begin{definition}
For symmetry element $\zee \in\CC[\SN]$, the {\em circular genome algebra} for $N$ regions is 
\[ \cA := \zee\CC[\SN] = \{\zee\bm{\tau} : \bm{\tau}\in\CC[\SN]\}\,. \]
Any element of $\cA$ of the form $\zee\sigma$, where $\sigma\in\SN$, is called a {\em permutation cloud}.
\end{definition}

One easily verifies that $\cA$ is a subalgebra of $\CC[\SN]$. To see that $\cA$ has a natural basis that is in correspondence with the set of genomes, firstly observe that any element of $\cA$ can be written as a linear combination of permutation clouds $\zee\sigma$, for $\sigma\in\SN$. Thus there exists a basis for $\cA$ of the form 
$\{ \zee\sigma_1 , \ldots , \zee\sigma_{K} \}$, for $\sigma_i\in\SN$.
Now, 
\[ \zee\sigma_i = \tfrac{1}{2N}\sum_{d\in\DN} d\sigma_i \,, \]
so that each basis element is a weighted sum of elements from a set $[\sigma_i]$, representing a particular genome. Since the sets are equivalence classes, for any $\sigma_1,\sigma_2\in\SN$ we have
\[ \zee \sigma_1 = \zee\sigma_2 \;\iff\; \sigma_1,\sigma_2\in [\sigma] \textrm{ for some } \sigma\in\SN\,.\]
This means that the set of distinct permutation clouds corresponds to the set of distinct genomes, and these form a basis for $\cA$. Finally, noting that for all $\sigma\in\SN$, $\big|[\sigma]\big| = 2N$, we see that
\begin{equation}
\dim(\cA) = \big|\{ [\sigma] : \sigma\in\SN \}\big| = \tfrac{N!}{2N} =: K\,.
\end{equation}
For the remainder of this section, we fix a basis for $\cA$,
\begin{equation}\label{eq:basis-A} B:= \{ \zee\sigma : \sigma\in\SN\} = \{ \zee\sigma_1 , \ldots , \zee\sigma_{K} \} \,,
\end{equation}
where we have chosen a representative $\sigma_i \in\SN$ of each equivalence class $[\sigma_i]$ for notational convenience. We set the first basis element to correspond to $[e] = \DN$, so that $\zee\sigma_1 = \zee$. 
Having used the symmetry of the genomes to construct the algebra, we now incorporate the symmetry of the model to extract some useful properties. 

\begin{proposition}\label{prop:comm.idem}
The model and symmetry elements, $\ess, \zee\in \CC[\SN]$, have the following properties. 
\begin{enumerate}
\item[(i)] $\zee$ is idempotent;
\item[(ii)] $\ess$ and $\zee$ commute.
\end{enumerate}
\end{proposition}
\begin{proof}
(i) Since $\DN$ is a group, we have
\[ \zee^2 = \tfrac{1}{(2N)^2}\sum_{d\in\DN}\sum_{f\in\DN} d\, f = \tfrac{1}{(2N)^2}\sum_{d\in\DN}\sum_{f\in\DN} f
			= \tfrac{1}{2N}\sum_{f\in\DN} f = \zee\,. \]
(ii) Now we use the dihedral symmetry (\ref{cond:m1}) of the model to rewrite the model as $m$ {\em base rearrangements}, $a_1,\ldots, a_m\in\SN$, along with their  symmetries. That is, 
\begin{equation}\label{eq:sym-model}
 \cM = \{  da_1 d^{-1}, da_2 d^{-1}, \ldots , da_m d^{-1} : d\in\DN\}\,.
\end{equation}
Then, using the same idea as in (i),
\begin{align*}
\zee\ess 		&= \tfrac{1}{2N}\sum_{f\in\DN}\sum_{i=1}^m \sum_{d\in\DN} w(a_i) f d a_i d^{-1} 
			= \tfrac{1}{2N}\sum_{f\in\DN}\sum_{i=1}^m \sum_{d\in\DN} w(a_i)  d a_i d^{-1} f
			= \ess\zee\,. \qedhere
\end{align*}
\end{proof}

The above properties translate immediately into properties of the representations of $\zee$ and $\ess$.

\begin{corollary}\label{cor:zee-ess}
Let $p\vdash N$ and $\rho_p:\CC[\SN]\to M_{D_p}(\CC)$ denote the corresponding irreducible representation of the symmetric group algebra. Then
\begin{enumerate}
\item[(i)] the only eigenvalues of $\rho_p(\zee)$ are $0$ and $1$;
\item[(ii)] $\rho_p(\zee)$ and $\rho_p(\ess)$ are simultaneously diagonalisable, with real eigenvectors. 
\end{enumerate}
\end{corollary}
\begin{proof}
Claim (i) is immediate since $\zee$, and thus $\rho_p(\zee)$, is idempotent. To show (ii), we firstly choose the representation $\rho_p$ to be orthogonal on $\SN$ \cite{sagan}. Then the rearrangement reversibility of the model ensures that $\rho_p(\ess)$ is symmetric \cite{vjez_circ1} and, similarly, one may verify directly that $\rho_p(\zee)^{\mathsf{T}}=\rho_p(\zee)$.
Thus, since $\zee$ and $\ess$ commute, the representation matrices commute and are simultaneously diagonalisable. In particular, since these matrices are real symmetric, the orthonormal set of simultaneous eigenvectors may be chosen to be real. 
\end{proof}

Now fix $p\vdash N$ and choose a set of orthonormal vectors $\{v_1,v_2,\ldots,v_{D_p}\}\subseteq\RR^{D_p}$ that are eigenvectors for both $\rho_p(\zee)$ and $\rho_p(\ess)$,  
ordered so that the first $k_p$ of them are eigenvectors for the eigenvalue $1$ of $\rho_p(\zee)$.
Take an eigenvalue, $\lambda_{p,i}$ of $\rho_p(\ess)$ and let $J_i\subseteq \{1,2,\ldots,D_p\}$ such that $\{v_j :j\in J_i\}$ are the eigenvectors for $\lambda_{p,i}$. Then, for $\sigma\in\SN$, the partial trace for $\lambda_{p,i}$ may be written as
\begin{equation}\label{eq:partials}
 \tr(\rho_p(\sigma^{-1}\zee) E_{p,i}) = \sum_{j\in J_i} v_j^T \rho_p(\sigma^{-1}\zee) v_j \,,
\end{equation}
where, for each $j$, 
\begin{equation}\label{eq:knockout}
 v_j^T \rho_p(\sigma^{-1}\zee) v_j = v_j^T \rho_p(\sigma^{-1})\rho_p(\zee) v_j  
= \left\{ \begin{array}{ll}
v_j^T \rho_p(\sigma^{-1}) v_j 	&, \textrm{ if } j\leq k_p\,;\\
0								&, \textrm{ if } j > k_p \,.\end{array}\right.
\end{equation}
We see then that $\rho_p(\zee)$ ``knocks out'' parts of the partial traces; in particular, it does this independently of the genome. 
We shall establish shortly that, in total, $\tfrac{2N-1}{2N}$ of the partial traces are knocked out in this way, thus explaining the observation in \cite{vjez_circ1} that most of the calculated partial traces were zero. 

The key to performing MLE computations in the symmetric group algebra is the relationship between the character of the regular representation and the identity element, $e\in\CC[\SN]$: the character $\chi_{\reg}(\bm{\tau})$ counts occurrences of the identity in a generic element $\bm{\tau}\in\CC[\SN]$. Since $\zee$ is idempotent, it is a left identity (but not a right identity) in the algebra $\cA$. We now construct the regular representation, $\rho^{\cA}_{\reg}$, of $\cA$ and show that its character, the {\em regular character} $\chi_{\reg}^{\cA}$, functions in exactly this way for the left identity $\zee\in\cA$. 

We construct the regular representation of $\cA$ via the left action of elements of $\cA$ on the basis $B=\{ \zee\sigma_1 , \ldots , \zee\sigma_{K} \}$ fixed above (\ref{eq:basis-A}).  
We need only consider the representation of a generic basis element, $\zee\sigma$ for $\sigma\in\SN$, since one may extend linearly to all of $\cA$. For arbitrary $\sigma\in\SN$, the $ij$th entry of the matrix $\rho^{\cA}_{\reg}(\zee\sigma)$ is the coefficient of $\zee\sigma_i$ in the expansion of $(\zee\sigma)(\zee\sigma_j)$,
that is, 
\begin{equation}\label{eq:regrepA}
\big(\rho^{\cA}_{\reg}(\zee\sigma)\big)_{ij} = \tfrac{1}{2N}\,  \big|\{ d\in \DN : \sigma d \sigma_j \in [\sigma_i] \}\big| \,.
\end{equation}
One readily verifies that $\rho^{\cA}_{\reg}(\zee)$ is the $K\times K$ identity matrix and that $\rho^{\cA}_{\reg}(\zee\sigma)^T = \rho^{\cA}_{\reg}(\zee\sigma^{-1})$. The regular character $\chi_{\reg}^{\cA}$ is the trace of the regular representation matrix. For a generic basis element $\zee\sigma\in\cA$,
\begin{align}
\chi^{\cA}_{\reg}(\zee\sigma) 
= \sum_{i=1}^{K} (\rho^{\cA}_{\reg}(\zee\sigma )_{ii} 
&= \tfrac{1}{2N} \sum_{i=1}^{K} \big|\{ d_1\in \DN : \sigma d_1 \sigma_i \in [\sigma_i] \}\big| \nonumber\\
&= \tfrac{1}{2N} \sum_{i=1}^{K} \sum_{d_2\in \DN} \big|\{ d_1\in \DN : \sigma d_1 \sigma_i =d_2\sigma_i \}\big| \nonumber\\
&= \tfrac{1}{2N} \sum_{i=1}^{K} \sum_{d_2\in \DN} \big|\{ d_1\in \DN : \sigma d_1 =d_2 \}\big| \nonumber\\
&= \tfrac{1}{2N} K \sum_{d_2\in \DN} \big|\{ d_1\in \DN : d_1 =\sigma^{-1} d_2 \}\big| \nonumber\\
&= \left\{ \begin{array}{ll} 
K 	& \textrm{ if } \sigma\in \DN \,; \\
0				& \textrm{ if } \sigma\not\in \DN \,.\end{array}\right.\label{eq:regcharA}
\end{align} 
Since $\zee\sigma = \zee$ if and only if $\sigma\in \DN$, this shows that we can use the character of the regular representation of the algebra $\cA$ to track coefficients of the left identity $\zee$, just as we do for the identity $e$ in $\CC[\SN]$. Further, we can express the regular character of $\cA$ as a sum over the irreducible characters of $\CC[\SN]$, and thus see that the regular characters of $\cA$ and $\CC[\SN]$ coincide on $\cA$.

\begin{proposition}\label{prop:regrepA}
For arbitrary $\bm{\tau}\in\cA$,
\begin{enumerate}
\item[(i)]  $\tfrac{1}{K}\,\chi^{\cA}_{\reg}(\bm{\tau})$ is the coefficient of $\zee$ in $\bm{\tau}\,;$
\item[(ii)] $\chi^{\cA}_{\reg}(\bm{\tau}) =\ds\sum_{p\vdash N} \chi_p(e) \chi_p(\bm{\tau}) 
									= \ds\sum_{p\vdash N} D_p \chi_p(\bm{\tau})=\chi_{\reg}(\bm{\tau})\,.$
\end{enumerate}
\end{proposition}
\begin{proof}
(i) Given $\bm{\tau}\in\cA$ and the basis $B$ from (\ref{eq:basis-A}), there exist $c_1,\ldots,c_K\in\CC$ such that 
$\bm{\tau} = c_1\zee  + c_2 \zee\sigma_2 + \ldots + c_K \zee\sigma_K$. Then
\[ \chi^{\cA}_{\reg}(\bm{\tau}) = c_1\chi^{\cA}_{\reg}(\zee) + \sum_{i=2}^K c_i\chi^{\cA}_{\reg}(\zee\sigma_i)
										= c_1 \, K \,,\]
from (\ref{eq:regcharA}).

(ii) It suffices to consider a generic basis element $\zee\sigma\in\cA$, since the characters are linear. We shall apply the dual orthogonality relations on the irreducible characters of $\SN$ (see for example \cite[Thm. 16.4]{James_Liebeck}), given for $\sigma,\tau\in \SN$ by
\begin{equation}\label{eq:orthogrels}
\sum_{p \vdash N} \chi_{p}(\sigma)\overline{\chi_p(\tau)} = \delta(\sigma,\tau) \left|\mathsf{cent}_{\SN}(\sigma)\right|\,,
\end{equation}
where $\mathsf{cent}_{\SN}(\sigma):=\{\gamma\in\SN : \gamma\sigma=\sigma\gamma\}$ is the centraliser of $\sigma$ and the map $\delta:\SN\to \{0,1\}$ is defined by
\begin{equation}\label{eq:delta equiv}
\delta(\sigma,\tau) = \left\{ \begin{array}{ll} 
1\,, & \textrm{ if } \sigma, \tau \textrm{ are in the same conjugacy class in } \SN ;\\
0\,, & \textrm{ otherwise}\,. \end{array}\right. 
\end{equation}
Recall that for $p\dashv N$ and $\tau\in\SN$, $\overline{\chi_p(\tau)} = \chi_p(\tau^{-1})$.\footnote{Recall that the irreducible characters of $\SN$ are all real valued, however we use the general form of the result here, since we apply the same argument to an arbitrary group in Section~\ref{sec:gen}.} Then for $\sigma\in\SN$, we have
\begin{align*}
\sum_{p\vdash N} \chi_p(e) \chi_p(\zee\sigma) 
	&=\tfrac{1}{2N}\sum_{d\in \DN}\sum_{p\vdash N} \chi_p(e) \overline{\chi_p(\sigma^{-1}d)} \\
	&=\tfrac{1}{2N}\sum_{d\in \DN}\delta(e,\sigma^{-1}d)\, \big|\mathsf{cent}_{\SN}(e) \big| \\
	&= \left\{ \begin{array}{ll} 
	\tfrac{N!}{2N} 	& \textrm{ if } \sigma\in \DN \,; \\
	0				& \textrm{ if } \sigma\not\in \DN \,.\end{array}\right. \\
	&= \chi_{\reg}^{\cA}(\zee\sigma)\,,
\end{align*}
by (\ref{eq:regcharA}), recalling that $K = \tfrac{N!}{2N}$.
\end{proof}

An immediate consequence of the above is an expression for the dimension of $\cA$ in terms of the characters $\chi_p$ of $\CC[\SN]$:
\begin{equation}\label{eq:dimA} 
\dim(\cA) = K = \tfrac{N!}{2N} = \sum_{p\vdash N} \chi_p(\e)\chi_p(\zee) = \sum_{p\vdash N} D_p k_p\,,
\end{equation}
where, for each $p\dashv N$, $k_p$ is the multiplicity of the eigenvalue $1$ of $\rho_p(\zee)$.
The first part of this can also be seen directly from the dual orthogonality relations.

To perform the MLE calculations in the algebra $\cA$ efficiently, we will need a decomposition of the regular character in terms of irreducible characters of $\cA$. 
Firstly, we'll observe that irreducible submodules of $\cA = \zee\CC[\SN]$ can be produced by acting with $\zee$ on the irreducible submodules of $\CC[\SN]$. This is straightforward and, in fact, true in a more general context (see, for example, \cite[Lemma 4.15]{steinberg_book}), however we include the details here since we'll use them in our subsequent constructions.
We denote the irreducible submodules of $\CC[\SN]$ by $V_p= \CC^{D_p}$ for $p\dashv N$, that is, we write 
\begin{equation}\label{eq:groupalg_decomp}
 \CC[\SN] \cong \bigoplus_{p\vdash N} D_p V_p\,. 
\end{equation}

\begin{theorem}\label{thm:Adecomp}
The non-trivial modules gained by acting with the symmetry element $\zee$ on the irreducible submodules of $\CC[\SN]$ are irreducible modules of the genome algebra $\cA$. 
\end{theorem}
\begin{proof}
Let $p\dashv N$. As above, we may take a set $\{  v_1, v_2 , \ldots ,  v_{D_p}\}$ of (real, orthonormal, linearly independent) eigenvectors for both $\rho_p(\ess)$ and $\rho_p(\zee)$, ordered such that the first $k_p$ of them correspond to the eigenvalue $1$ of $\rho_p(\zee)$. Then $V_p =  \mathsf{span}_{\CC}\{ v_i : i = 1, \ldots , D_p\}$ and 
\begin{equation}\label{eq:Wpmodules} 
W_p:= \zee\cdot V_p = \mathsf{span}_{\CC}\{ \rho(\zee) v_i : i = 1, \ldots , D_p\} = \mathsf{span}_{\CC}\{ v_i : i = 1, \ldots , k_p\}\,.
\end{equation}
It is clear that $W_p$ is an $\cA$-module; we need show that it is either $\{0\}$ or irreducible. Suppose that there exists $U_p\subseteq W_p$ such that $U_p$ is an $\cA$-module, and $0\neq u\in U_p$. 
Then
\[ U_p \;  \supseteq \; \mathsf{span}_{\CC}\{ \rho_p(\zee\sigma)u : \sigma\in\SN\}
= \zee\cdot \mathsf{span}_{\CC}\{\rho_p(\sigma)u : \sigma\in\SN \} 
= \zee \cdot V_p 
= W_p \,, 
 \]
so that $U_p = W_p$.
\end{proof}
 
It was shown in \cite{vjez_circ1} that there exist $p\dashv N$ for all $N>3$ such that $\chi_p(\zee) =k_p = 0$; that is, there are always some $\CC[\SN]$-modules that are projected down to zero in $\cA$. We note that this is not true in the more general case considered in Section~\ref{sec:gen} (for example if $\zee$ is constructed from a different symmetry group). 

Now, the dimension expression (\ref{eq:dimA}) suggests that we will not be able to decompose the algebra $\cA$ into a direct sum of irreducible submodules as we can for $\CC[\SN]$ (\ref{eq:groupalg_decomp}). By \cite[Thm. 3.5.8]{etingof_book}, if this were possible with the irreducible submodules $W_p$ from above, then the dimension of $\cA$ would be $\sum_{p\dashv N} k_p^2$. We have not yet verified here that these $W_p$ comprise {\em all} irreducible submodules of $\cA$, nor that they are all distinct (not isomorphic to one another), but this is indeed the case \cite[Thm. 4.23]{steinberg_book}. The difference between the dimension of $\cA$ and that gained from the irreducible modules here is signalling that not all of the information about $\cA$ can be represented by the action of $\cA$ --- in this case, the left action of $\cA$ on its irreducible modules, and on itself, is not injective. 

To see this, let $W$ be an irreducible module of $\cA$. Then, since $\zee$ is a left identity in $\cA$, $\zee$ must act as the identity on $W$. But then, for any $\zee\sigma, \zee\sigma^{\prime}\in\cA$ such that $\zee\sigma\zee = \zee\sigma^{\prime}\zee$, 
\begin{equation}\label{eq:notinjective} 
(\zee\sigma)\cdot w =  (\zee\sigma)\cdot(\zee\cdot w) = (\zee\sigma\zee)\cdot w = (\zee\sigma^{\prime}\zee)\cdot w = (\zee\sigma^{\prime})\cdot w \,,
\end{equation}
for all $w\in W$. 

From Theorem \ref{thm:equivs} (ii), such $\zee\sigma\neq\zee\sigma^{\prime}\in\cA$ correspond to physically distinct genomes that share the same path probabilities and likelihood functions: we have $\zee\sigma\zee = \zee\sigma^{\prime}\zee$ if and only if $\sigma^{\prime} \in[\sigma]_D = \{d\sigma d^{\prime} : d,d^{\prime}\in\DN\}$.   

In the language of algebras, $\cA$ has a {\em non-trivial radical}  \cite[Def. 3.5.1]{etingof_book}, since (for $N>3$) there are non-zero elements $\zee\sigma - \zee\sigma^{\prime}\in\cA$ that annihilate all irreducible modules of $\cA$. As a concrete example, consider the following.

\bt{\begin{example}\label{eg:zsigz}
Let $\sigma = (1,2), \sigma^{\prime} = (2,3) \in\SN$. Setting $r=(1,2,\ldots,N)\in\DN$, we observe that $\sigma^{\prime} = r\sigma r^{-1}$, so that $\zee\sigma\zee = \zee\sigma^{\prime}\zee$. But there exists no $d\in\DN$ such that $\sigma = d\sigma^{\prime}$, and thus $\zee\sigma\neq\zee\sigma^{\prime}$. \exend
\end{example}}

For our practical purposes, this is perfect: the algebra sees genomes as distinct entities, but their representations do not distinguish between genomes corresponding to an equivalence class $[\sigma]_D$, whose likelihood functions are the same. \bt{Further, whilst $\zee\sigma$ and $\zee\sigma^{\prime}$ correspond to distinct genomes, if we consider them as rearrangements, they are not distinct, since they have the same action. We shall return to this presently, when we define models in the genome algebra.}

\begin{proposition}\label{prop:regAnot1-1}
Let $\zee\sigma,\zee\sigma^{\prime}\in\cA$. Then $\rho_{\reg}^{\cA}(\zee\sigma)=\rho_{\reg}^{\cA}(\zee\sigma^{\prime})$ if and only if $\sigma^{\prime}\in[\sigma]_D$.
\end{proposition}
\begin{proof}
For the reverse implication, we argue as above, replacing $w$ in (\ref{eq:notinjective}) by each basis element $\zee\sigma_i$ to verify that the matrices are the same. Conversely, if the regular representations coincide, then we immediately have that $(\zee\sigma)\zee = (\zee\sigma^{\prime})\zee$.
\end{proof}

We now explicitly consider the irreducible representations of $\cA$ on the irreducible submodules and use these to rewrite the regular character of $\cA$ in terms of irreducible characters of $\cA$. Let $p\dashv N$ such that $k_p>0$ and consider the module $W_p$ of $\cA$ which, as in the proof of Theorem \ref{thm:Adecomp}, has a basis $\{v_1,\ldots , v_{k_p}\}\subseteq\RR^d$ of orthonormal eigenvectors. Now, the action of $\cA$ on $W_p = \zee\cdot V_p$ is inherited from the action of $\CC[\SN]$ on $V_p$, so for arbitrary $\bm{\tau}\in\cA$, we define the $(k_p\times k_p)$ representation matrix $\rho^{\cA}_p(\bm{\tau})$ on $W_p$ via the action of $\rho_p(\bm{\tau})$ on the basis vectors $v_j$:
\[ \left(\rho^{\cA}_p(\bm{\tau})\right)_{ij} := v_i^T \rho_p(\bm{\tau})v_j \,.\] 
More concisely, setting $Q_p$ to be the $(D_p \times k_p)$ matrix with $\{v_1,\ldots , v_{k_p}\}$ as columns, we have
\begin{equation}\label{eq:QrepA} \rho^{\cA}_p(\bm{\tau}) = Q_p^T \rho_p(\bm{\tau}) Q_p \,.
\end{equation}

Clearly (also c.f. (\ref{eq:notinjective})), $\rho_p^{\cA}(\zee)$ is the $(k_p\times k_p)$ identity matrix for each such $p\dashv N$. For each $p\dashv N$ such that $k_p=0$, we formally define $\rho_p^{\cA}$ to be the zero representation. 
Now we may calculate the irreducible characters $\chi^{\cA}_p$ of $\cA$ and see that they coincide with the irreducible characters $\chi_p$ of $\CC[\SN]$ restricted to $\cA$. 

\begin{proposition}\label{prop:chars_equal}
For each $p \vdash N$ and $\bm{\tau}\in\cA$, $\chi^{\cA}_p(\bm{\tau}) = \chi_p(\bm{\tau})$.
\end{proposition}
\begin{proof}
Let $p\dashv N$. By linearity, we need only verify the claim on a generic basis element, $\zee\sigma\in\cA$. Again utilising the orthonormal eigenvectors $\{v_1,\ldots,v_{D_p}\}$ of $\rho_p(\zee)$, where those for $i\leq k_p$ correspond to the eigenvalue $1$ and the remainder to the eigenvalue $0$, we have  
\begin{align*}
 \chi_p^{\cA}(\zee\sigma) =  \sum_{i = 1}^{k_p} \left(\rho^{\cA}_p(\zee\sigma)\right)_{ii} 
	= \sum_{i=1}^{k_p} v_i^T \rho_p(\zee\sigma)v_i 
	&= \sum_{i=1}^{D_p} v_i^T \rho_p(\zee\sigma\zee)v_i 
	= \chi_p(\zee\sigma\zee) 
	= \chi_p(\zee\sigma)  \,,
\end{align*}
where in the final step we have used the cyclicity of the trace and the idempotency of $\zee$ (Proposition~\ref{prop:comm.idem}).
\end{proof}

Combining Propositions \ref{prop:regrepA} and \ref{prop:chars_equal} gives the desired character decomposition.

\begin{corollary}\label{coro:Acharsdecomp}
For arbitrary $\bm{\tau}\in\cA$,
\[ \chi^{\cA}_{\reg}(\bm{\tau}) =\sum_{p\vdash N} D_p \chi_p^{\cA}(\bm{\tau})\,.   \] 
\end{corollary}
\vspace{-1cm} \qed

Having defined and decomposed the regular character of $\cA$, we are ready to return to the likelihood calculations. Using the equivalence of the characters of $\cA$ and $\CC[\SN]$ on the algebra $\cA$, along with the the interplay between the genome and model symmetry, we now verify that we may work entirely in $\cA$ to calculate the genome path probabilities and thus the likelihood functions, as defined in the previous section (\ref{eq:simplelike}). 

\begin{theorem}\label{thm:alphasA}
Let  $\sigma\in\SN$ and $k\in \NN_0$. Then
\begin{equation}\label{eq:alpha-char-A}  \alpha_k(\sigma) = \tfrac{2N}{N!}\chi^{\cA}_{\reg}(\zee\sigma^{-1}\zee\ess^k) 
			= \tfrac{2N}{N!} \sum_{p\vdash N} D_p\, \chi^{\cA}_{p}(\zee\sigma^{-1}\zee\ess^k)\,.
\end{equation}
\end{theorem}
\begin{proof}
From (\ref{eq:alpha_k-SN}), $\alpha_k(\sigma)  =  \tfrac{2N}{N!} \chi_{\reg}(\sigma^{-1}\zee\ess^k)=  \tfrac{2N}{N!} \chi_{\reg}(\zee\sigma^{-1}\zee\ess^k)$, since $\zee$ and $\ess$ commute, $\zee$ is idempotent and the trace is cyclic. The first equality is then clear from Proposition~\ref{prop:regrepA} and the second from Corollary~\ref{coro:Acharsdecomp}.
\end{proof}

We have mentioned the importance of the `identity counting' property of the regular character, that is, Proposition~\ref{prop:regrepA} (i), but this combinatorial component is somewhat hidden in the proof of Theorem~\ref{thm:alphasA}. To highlight it, one may begin with the identity (\ref{eq:s^k}) stated in the previous section and, for any given genome $\zee\sigma$ ($\sigma\in\SN$), multiply by $\zee\sigma^{-1}\zee$ to obtain
\[ \zee\sigma^{-1}\zee\ess^k=\tfrac{1}{2N}\sum_{\tau\in\SN} \alpha_k(\tau)\,\zee\sigma^{-1}\tau\,.\]
By observing that there are exactly $2N$ values of $\tau\in\SN$ for which $\zee\sigma^{-1}\tau=\zee$, one thus sees directly that the coefficient of $\zee$ in the expansion of $\zee\sigma^{-1}\zee\ess^k$ is $\alpha_k(\sigma)$.

Note that we could have simplified the above character expression (\ref{eq:alpha-char-A}) a little, that is,  
\[ \chi_p^{\cA}(\zee\sigma^{-1}\zee \ess^k) =  \chi_p^{\cA}(\zee\sigma^{-1}\ess^k \zee) =  \chi_p^{\cA}(\zee\sigma^{-1}\ess^k)\,.\]
However, as we did in the algebra $\CC[\SN]$, we want to diagonalise the matrices representing the model element, namely the matrices $\rho_p^{\cA}(\zee\ess)$. So we keep the middle $\zee$ and write
\[ \chi_p^{\cA}(\zee\sigma^{-1}\zee \ess^k)  = \tr\left(\rho_p^{\cA}\big(\zee\sigma^{-1}(\zee \ess)^k\big)\right)
											= \tr\left(\rho_p^{\cA}(\zee\sigma^{-1})\rho_p^{\cA}(\zee \ess)^k\right)\,.\]
For each $p\dashv N$, as in the proof of Corollary \ref{cor:zee-ess}, we can choose $\rho_p(\ess)$ to be symmetric; thus by the definition (\ref{eq:QrepA}) each matrix $\rho^{\cA}_p(\zee\ess)$ is symmetric and thus diagonalisable. Then we obtain
\begin{equation}\label{eq:ppsA}	
\alpha_k(\sigma) = \tfrac{2N}{N!} \sum_{p\vdash N} D_p \sum_{i=1}^{R_p} 
			\lambda_{p,i}^k \mathrm{tr}(\rho_p^{\cA}(\zee\sigma^{-1}) E^{\cA}_{p,i})\,,
\end{equation}
where $E^{\cA}_{p,i}$ is the projection onto the eigenspace of the $i$th eigenvalue, $\lambda_{p,i}$, of $\rho_p^{\cA}(\zee\ess)$. 

Now, finally substituting the path probabilities (\ref{eq:ppsA}) into the theoretical likelihood expression (\ref{eq:simplelike}), we obtain
\begin{equation}\label{eq:Alike}
L(T|\sigma) = 
e^{-T} \tfrac{2N}{N!}\sum_{p\vdash N} D_p\sum_{i=1}^{R_p}\tr(\rho^{\cA}_p(\zee\sigma^{-1}) E^{\cA}_{p,i}) e^{\lambda_{p,i}T}\,.
\end{equation}

It is clear from Theorem~\ref{thm:alphasA} that the likelihood expression (\ref{eq:Alike}), involving only elements of the genome algebra $\cA$, is equal to that (\ref{eq:perms-like}) gained via the group algebra $\CC[\SN]$. We now show that the above is, really, a simplified version of (\ref{eq:perms-like}): that is, by working in the smaller algebra we have eliminated the many eigenvalue terms that occur with zero coefficients. 

\begin{proposition}\label{prop:eqevals}
For each $p\vdash N$ such that $W_p\neq \{0\}$, the eigenvalues of the matrix $\rho^{\cA}_p(\zee\ess)$ are exactly the eigenvalues of $\rho_p(\ess)$ that occur with non-zero coefficient in the likelihood expression (\ref{eq:perms-like}).
\end{proposition}
\begin{proof}
Let $p\dashv N$ such that $W_p\neq\{0\}$. As above, take the set $\{v_1,v_2,\ldots,v_{D_p}\}\subseteq\RR^{D_p}$ of orthonormal eigenvectors for both $\rho_p(\zee)$ and $\rho_p(\ess)$, with the first $k_p$ corresponding to the eigenvalue $1$ of $\rho_p(\zee)$, and form the matrix $Q_p$ with the first $k_p$ vectors as columns. Then, as in (\ref{eq:QrepA}),
\begin{equation}\label{eq:evalue_matrix}
\rho_p^{\mathcal{A}}(\zee\ess)	
	= Q_p^T \rho_p(\zee\ess) Q_p 
	= Q_p^T \rho_p(\ess\zee) Q_p 
	= Q_p^T \rho_p(\ess) Q_p 
	= \left( \begin{array}{cccc} 
				\lambda_1 & 0 & \ldots & 0 \\
				0 & \lambda_2 & \ldots & 0 \\
				& & \ddots &  \\
				0 & 0 & \ldots & \lambda_{k_p} \end{array}\right) \,, 
\end{equation}
where each $\lambda_i$ is clearly an eigenvalue of both $\rho^{\cA}_p(\zee\ess)$ and $\rho_p(\ess)$ (and the $\lambda_i$ are not necessarily distinct). Suppose $\lambda^{\prime}$ is an eigenvalue of $\rho_p(\ess)$ that does not appear in the matrix of (\ref{eq:evalue_matrix}). Then $\lambda^\prime$ has corresponding eigenvector(s) $\{v_j: j\in J^{\prime}\}$, where $J^{\prime}\subseteq \{k_{p+1}, k_{p+2}, \ldots, D_p\}$.
But then, for any $\sigma\in\SN$, the coefficient of the $\lambda^{\prime}$ term in the likelihood expression is the partial trace
\[ \sum_{j\in J^{\prime}} v_j^T \rho_p(\sigma^{-1}\zee) v_j = 0\,,\]
by (\ref{eq:partials}) and (\ref{eq:knockout}). 
\end{proof}

Note that, although in the proof of Proposition \ref{prop:eqevals} we construct each $\rho_p^{\cA}(\zee\ess)$ as a diagonal matrix (in which case the projections onto the eigenspaces would be diagonal matrices of $1$s and $0$s), we do this only to verify that the representation has the required properties, and we utilise the eigenvectors of the representation $\rho_p(\ess)$. In practice, the whole point is to {\em not} calculate the much bigger representations $\rho_p(\ess)$. That is, when implementing calculations, we would expect to construct a basis for each irreducible module $W_p$ directly, hence the general form of the projections in (\ref{eq:Alike}). 

We note that the equivalence of the path probabilities and thus likelihoods on the classes $[\sigma]_D$ and $[\sigma]_{DR}$ stated in Theorem \ref{thm:equivs} can alternatively be obtained by working directly in the genome algebra $\cA$. We omit the proof here since we shall prove a more general version of the result in Section~\ref{sec:gen}.

Since $\dim(\CC[\SN]) = N! = \ds\sum_{p\vdash N} D_p^2$ and $\dim(\cA) = \tfrac{N!}{2N} = \ds\sum_{p\vdash N} D_p k_p$, we have
\begin{equation}\label{eq:dimred}
 \sum_{p \vdash N} D_p k_p = \sum_{p\vdash N} D_p \tfrac{D_p}{2N} \,.
\end{equation}
\bt{Note that this does not imply that $k_p = \tfrac{1}{2N} D_p$ for each (or any) $p\vdash N$, rather that} {\em on average}, and asymptotically, the dimension of each irreducible submodule $W_p$ of $\cA$ is $\tfrac{1}{2N}$th of the dimension of the irreducible submodule $V_p$ of $\CC[\SN]$. To put this another way, on average, $\tfrac{2N-1}{2N}$ths of the computations in the group algebra, as documented in \cite{vjez_circ1}, resulted in zeroes. 

Given that the dimension of the algebra $\cA$ is still $\tfrac{N!}{2N}$, this does not significantly reduce the computational complexity. \bt{However, since the multiplicity of the irreducible submodules in $\cA$ is the same as in the group algebra (\ref{eq:dimred}), the reduction in the dimension of the irreducible submodules is (relatively) much larger than the reduction in total dimension.}

\bt{\begin{example}
Consider $N = 6$. There are $N! = 720$ permutations in $\mathcal{S}_6$, so the dimension of the regular representation of $\CC[\mathcal{S}_6]$ is 720. The dimensions of the irreducible modules $V_p$ of $\CC[\mathcal{S}_6]$ (given as a list rather than a set as they are not all distinct) are
\[ [D_p : p\vdash 6] = [1,5,9,10,5,16,10,5,9,5,1]\,.\]
Moving to the genome algebra, there are $\tfrac{N!}{2N} = 60$ distinct genomes, so the dimension of the regular representation of $\cA$ is $60$. The dimensions of the corresponding irreducible modules $W_p = \zee\cdot V_p$ of $\cA$ are 
\[ [k_p: p\vdash 6] = [1,0,2,0,0,1,1,2,0,1,0]\,.\]
Thus, for any rearrangement model, each likelihood expression will be a sum of at most eight terms, corresponding to at most eight distinct eigenvalues. \exend
\end{example}}

\bt{We note that such dimension reductions are less striking for larger $N$. In any case,} we see a significant theoretical gain here: the genome algebra $\cA$ incorporates the symmetry of the genomes and models into a unified framework, within which the problem can be formulated and the computations performed. To highlight this, we next consider the regular representations of $\ess$ in the algebra $\CC[\SN]$ and of $\zee\ess$ in the algebra $\cA$ as Markov matrices; then we conclude this section by re-formulating the model in the genome algebra framework. 

\subsection{The Markov interpretation}\label{subsec:markov}

In the group algebra, $\CC[\SN]$, the rows and columns of the regular representation are determined by the $N!$ permutations $\sigma_i\in\SN$. In particular, 
\[ \rho_{\reg}(\ess) = \sum_{a\in\cM} w(a) \rho_{\reg}(a)\,,\]
where for each rearrangement permutation $a\in\cM$, the $ij$th entry of $\rho_{\reg}(a)$ is $1$ if $a\sigma_j = \sigma_i$ and $0$ otherwise, so that the $\rho_{\reg}(a)$ matrices have exactly one `$1$' in each row and column. Then $\rho_{\reg}(\ess)$, as a convex sum of Markov matrices, is itself a Markov matrix. 
The $j$th column of $\rho_{\reg}(\ess)$ contains $|\cM|$ non-zero entries, each equal to a unique $w(a)$, since for the distinct permutations $a\in\M$, the permutations $a\sigma_j$ are all distinct.

Thus $\rho_{\reg}(\ess)$ is the transition matrix of a discrete Markov chain where the states are the $N!$  permutations in $\SN$ and the $ij$th entry is the probability of permutation $\sigma_j$ transitioning into permutation $\sigma_i$ via one rearrangement chosen from the model $\M$. That is, 
\[ \rho_{\reg}(\ess)_{ij} = \left\{\begin{array}{ll} w(a) & \textrm{ if } a= \sigma_i\sigma_j^{-1}\in\cM \,,\\ 0 & \textrm{ otherwise}\,. \end{array} \right. \]
It is clear from this formulation that the matrix $\rho_{\reg}(\ess)$ is symmetric if and only if the model has the rearrangement reversibility property (\ref{cond:m2}). Thus, since the stationary distribution on the Markov chain is the uniform distribution on the states, the reversibility property (\ref{cond:m2}) of the model is equivalent to reversibility of the Markov model. 

Now, in the algebra $\cA$, the corresponding matrix representing the model element is
\begin{equation}\label{eq:reg-zs} 
\rho^{\cA}_{\reg}(\zee\ess) = \sum_{a\in\cM} w(a) \rho^{\cA}_{\reg}(\zee a)\,.
\end{equation}
As above, the matrices on the right hand side represent basis elements of the algebra, here $\zee a$ for $a\in\M$. 
Although the basis elements here do not form a group (so their regular representations are not, in general, zero-one matrices), each of the $\rho^{\cA}_{\reg}(\zee a)$ is again a Markov matrix: for a given $a\in\cM$, the $ij$th entry of $\rho^{\cA}_{\reg}(\zee a)$ is the coefficient of $\zee\sigma_i$ in the expansion of $(\zee a) (\zee\sigma_j)$ and, since $\zee = \tfrac{1}{2N}\sum_{d\in\DN}d$, the expansion is a convex sum. Thus the entries in each column of each $\rho^{\cA}_{\reg}(\zee a)$ sum to one and $\rho^{\cA}_{\reg}(\zee\ess)$, as a convex sum of Markov matrices, is indeed a Markov matrix.

Each basis element $\zee\sigma_i$ corresponds to a genome, so $\rho^{\cA}_{\reg}(\zee\ess)$ is the transition matrix of a Markov chain where the states are genomes. The $ij$th entry, which we calculated as the proportion of the expansion of $(\zee\ess) (\zee\sigma_j)$ that is equal to $\zee\sigma_i$, is of course the probability of the genome $ (\zee\sigma_j)$ transitioning into the genome $\zee\sigma_i$ in one step, via the model.

\subsection{Permutation clouds: a unifying concept}\label{subsec:clouds}

In the permutation approach detailed in Section~\ref{sec:perms}, we considered rearrangement events to be individual permutations acting on individual permutations. In the setting of the genome algebra, we represent both genomes and rearrangements by permutation clouds, each of which is a sum of permutations weighted by their probabilities ($\zee\sigma = \tfrac{1}{2N}\sum_{d\in\DN} d\sigma$). A single rearrangement event is here modelled by a permutation cloud $\zee a$ acting on a permutation cloud $\zee\sigma$. Mathematically, this event results in a convex combination of permutation clouds $\sum c_i\zee\sigma_i$; biologically, it results in one of the genomes $\zee\sigma_i$, according to the probability distribution given by the coefficients $c_i$.  

The permutation cloud view of circular genomes seems to us quite natural. To observe a genome, we fix an orientation and a reference
frame, and assign to it a single permutation (any one, from the appropriate equivalence class $[\sigma]$, with probability $\tfrac{1}{2N}$). We refer to this as an {\em instance} of the genome. Theoretically, however, the genome exists simultaneously as all of its possible physical orientations in space; it {\em is} the cloud, $\zee\sigma$. 

What about rearrangements? For a rearrangement permutation $a\in\SN$ and $d\in\DN$, the result of the action $da$ on $\sigma\in\SN$ is $d(a\sigma)\in[a\sigma]$, that is, it results in the same genome as $a$ acting on $\sigma$. So we can think of $\zee a$ acting on $\zee \sigma$ as encompassing (all orientations of ($a$ acting on (all orientations of $\sigma$))).

Of course, the action of $\zee a$ on $\zee \sigma$ also incorporates the dihedral symmetries of $a$ as an action, that is, $dad^{-1}$ for $d\in\DN$. 
For a biological model $(\M,w,dist)$ for evolution of genomes as permutations, under the assumption of dihedral symmetry, we wrote (\ref{eq:sym-model})
\begin{equation}\label{eq:model-set} \M = \{ da_1 d^{-1}, da_2 d^{-1}, \ldots , da_m d^{-1} : d\in\DN\}\subseteq\SN \,,\end{equation}
where for each $a_k$ and all $d\in\DN$, $w(da_k d^{-1}) = w(a_k)$. 
Since $da_kd^{-1}\in [a_k]_D$, the action of $\zee (da_kd^{-1})$ on $\zee\sigma$ is the same as the action of $\zee a_k$ on $\zee\sigma$ (see Proposition \ref{prop:regAnot1-1}) and thus each $\rho_{\reg}^{\cA}(\zee (da_kd^{-1}))=\rho_{\reg}^{\cA}(\zee a_k)$.

Having shown that the MLE computations can be performed in the genome algebra $\cA$, and discussed the representation of both genomes and rearrangements as permutation clouds in this algebra, it remains to reformulate the model within this framework.
Given the model (\ref{eq:model-set}) in the permutation framework, the equivalent model for evolution in the genome algebra setting is $(\M^{\cA}, w^{\cA},dist)$, where
\[ \M^{\cA} := \{ \zee a_1 , \ldots , \zee a_m \}\subseteq\cA\,,\]
and $w^{\cA}(\zee a_i) = 2N w(a_i)$ for each $i$. 

Since the dihedral symmetry of the genomes is built into the algebra $\cA$, specifying the model to consist of elements of $\cA$ in this way makes the dihedral symmetry requirement (\ref{cond:m1}) redundant. \bt{Model reversibility in this setting is formulated as}
\begin{equation}
\textrm{ for each } \zee a\!\in\!\M^{\cA} , \;\zee a^{-1} \!\in\!\M^{\cA} \textrm{ and } w^{\cA}(\zee a^{-1}) = w^{\cA}(\zee a)\,.\tag{M2$^{\cA}$}\label{cond:m2A}
\end{equation}
\bt{This condition is sufficient to ensure that the irreducible representations of $\zee\ess$ are diagonalisable, which is convenient for computations.} 
Although the algebra $\cA$ has a left identity, it does not contain inverses, so $\zee a^{-1}$ is not (in general) an inverse of $\zee a$. However, as we shall see in the next section, \bt{model reversibility is, further, equivalent to the reversibility of the Markov model. We conclude this section with an example to illustrate some of these key concepts.}

\bt{\begin{example}\label{eg:models}
Suppose we wish to consider a model consisting only of ``small inversions'', which we will take to be inversions of two or three regions. In the permutation framework, we would define this model to be
\begin{align*}
\cM:&= \{ d (1,2) d^{-1}, d(1,3)d^{-1} : d\in\DN\} \\
	&= \{ (1,2),(2,3),\ldots ,(N-1,N), (N,1), (1,3),(2,4),\ldots,(N-1,1),(N,2)\}\,.
\end{align*}
Here there are $N$, rather than $2N$, distinct instances of each rearrangement type, since for inversions\footnote{\bt{in this non-oriented region case}}, each flip coincides with a rotation. 

For the rearrangement probabilities, one could choose the uniform distribution, $w(a) = \tfrac{1}{2N}$ for all $a\in \cM$, or one may consider the larger inversions to be less likely and set, for all $a\in\cM$,
\[ w_{\prime}(a)= \left\{ \begin{array}{ll}  \tfrac{2}{3N}, & \textrm{if } a =  d(1,2)d^{-1}, \textrm{ some } d\in\DN;\\
						\tfrac{1}{3N}, & \textrm{if } a =  d(1,3)d^{-1}, \textrm{ some } d\in\DN\,. \end{array}\right.\]
In the genome algebra, the model is simpler to express; we take the rearrangement instances $(1,2)$ and $(1,3)$ and the model is
\[\cM^{\cA}:= \{\zee(1,2),\zee(1,3)\}\,.\]
The weight functions corresponding to the above would then be $w^{\cA}(\zee(1,2)) = w^{\cA}(\zee(1,3)) = \tfrac{1}{2}$ or $w_{\prime}^{\cA}(\zee(1,2)) = \tfrac{2}{3}, w_{\prime}^{\cA}(\zee(1,3)) = \tfrac{1}{3}$.

One may recall from Example~\ref{eg:zsigz} that $\zee(1,2)\neq\zee(2,3)$, however, one need not (and indeed should not) include both of these in the rearrangement model since they have the same action: $\zee(1,2)\cdot\zee\sigma = \zee(2,3)\cdot\zee\sigma$ for any genome $\zee\sigma\in\cA$, since $\zee(1,2)\zee =\zee(2,3)\zee$. Further, one must be aware of {\em complementary rearrangements}. 

For example, in the case $N=5$, the actions of $\zee(1,2)$ and $\zee(1,3)$ coincide: inverting a two-region segment is, under dihedral symmetry, the same rearrangement as inverting the complementary three region segment (correspondingly, when $N=5$, $\zee(1,2)\zee=\zee(1,3)\zee$).\exend
\end{example}

One can easily eliminate the possibility of such `rearrangement redundancies'  by reformulating the model in the genome algebra as a set of elements of the form  $\zee a \zee$; we do this in the next section (\ref{eq:model-A}).  More on such considerations, along with explicit examples of rearrangement models in the oriented region case, may be found in \cite{vjoshjez}; a deeper algebraic consideration of rearrangements is given in \cite{joshvjez}. 
}

%%%%%%%%%%%%%%%%%%%%%%%%%%

\section{More general models of genomes}\label{sec:gen}

The construction of the genome algebra $\cA=\zee\CC[\SN]$ in Section~\ref{sec:A} was determined by assumptions we made about how to model the genomes. In particular, following on from previous work \cite{mles,jezandpet,vjez_circ1}, we chose to model circular genomes, without considering orientation of regions, which meant an instance of the genome could be represented by a permutation $\sigma\in\SN$. We \bt{modelled the genomes without a distinguished position}, which meant the genome symmetries corresponded to the dihedral group $\DN$. In this section, we outline how the constructions and techniques presented for this specific case can be generalised to cover different genomic models\bt{: any for which genome (and rearrangement) instances can be represented as elements of a group $G$}.

Suppose, for example, that one wanted to vary the above model to include an origin of replication in the circular genomes. We would model this as a distinguished position and the genomes would then have no rotational symmetry, only reflectional. The symmetry group would thus be $Z_N = \{e,f\}$, the symmetry element $\zee=\tfrac{1}{2}(e+f)$, and each genome an element $\zee\sigma = \tfrac{1}{2}(\sigma+f\sigma) \in\zee\CC[\SN]$. The model would naturally reflect this symmetry, with rearrangements taking the form $\zee a$, $a\in\SN$. \bt{In particular, this case allows for rearrangements at different positions on the genome, relative to the origin of replication, to be assigned different probabilities.}
With appropriate choices of rearrangements, this framework could also be used to represent linear genomes. 

To include orientation of genes, one would use a different underlying group, for example the hyperoctahedral group $H_N$ of signed permutations (as outlined in \cite{attilaand}), and a symmetry group of choice (for example, a copy of the dihedral group in the case of a circular genome with no distinguished positions).  \bt{An explicit consideration of the genome algebra for the signed region case, including some detailed examples, may be found in \cite{vjoshjez}.}

To construct the general genome algebra, we begin with a group $G$, whose elements represent instances of the genomes of interest, and a subgroup $Z\subseteq G$ that represents the physical symmetries of these genomes.\footnote{\bt{To model genomes as possessing no symmetries, one takes the symmetry group to be trivial; this case thus fits within the framework, but each genome simply corresponds to a single permutation.}} We consider \bt{rearrangements such that a single rearrangement event} for an instance $g\in G$ of a genome can be modelled via the left action of a particular element $a\in G$ on $g$.  
The terms in the following definition reflect our applications of the objects, but obviously the subsequent results concerning the algebras hold whether or not one applies them to genomes.

\begin{definition}
Let $G$ be a finite group with subgroup $Z\subseteq G$. Define
\[ \zee:= \tfrac{1}{|Z|}\ds\sum_{z\in Z} z\,, \quad \cA:= \zee\CC[G] \quad \textrm{ and } \quad \cA_0:= \zee\CC[G]\zee\,.\]
We call $\cA$ the {\em genome algebra of $G$ with $Z$}, $\cA_0$ the {\em class algebra of $G$ with $Z$} and $\zee$ the {\em symmetry element} of $\cA$ and $\cA_0$. 
\end{definition}

Rather than proceeding as in Sections~\ref{sec:perms} and \ref{sec:A}, where we first defined the rearrangement model, path probabilities and likelihoods for genome instances (group elements) and then showed that the calculations could be performed in the genome algebra, we will here formulate these concepts (and then perform the computations) entirely in the genome algebra $\cA$. We include the class algebra $\cA_0$ for completeness. Following the observations in the previous section, it seems a natural next step to consider the algebra formed by combining together the elements of $\cA$ that act indistinguishably. However, we shall see that this lower dimensional algebra is not the appropriate setting for our calculations. 

\begin{lemma}\label{lem:Z-G-classes}
Let $G$ be a finite group with subgroup $Z\subseteq G$. 
\begin{enumerate}
\item[(i)] For each $g\in G$, define $[g]: = \{zg : z\in Z\}$. Then the sets $\{[g] : g\in G\}$ are equivalence classes of $G$. For each $g\in G$, $\left|[g]\right| = |Z|$.
\item[(ii)] For each $g\in G$, define $[g]_D: = \{zgz^{\prime} : z,z^{\prime}\in Z\}$. Then the sets $\{[g]_D : g\in G\}$ are equivalence classes of $G$.
\end{enumerate}
\end{lemma}
\begin{proof}
Since the sets $[g]$ and $[g]_D$ for $g\in G$ are respectively right cosets and double cosets of $G$ with respect to the subgroup $Z$, it is clear that they are equivalence classes.
\end{proof}

We use the label `$D$' for the classes defined in (ii) above to refer to the double coset structure of the sets $[g]_D$ (noting that this conveniently coincides with the original usage \cite{vjez_circ1} of the label, which referred to the dihedral symmetry in the circular genome case). 
The following statements can be derived directly from the subgroup properties of $Z$ and Lemma \ref{lem:Z-G-classes} (c.f. the corresponding results in Section \ref{sec:A}).

\begin{proposition}\label{prop:gen-bases}
Let $G$ be a finite group with subgroup $Z\subseteq G$. Let $\cA$ and $\cA_0$ respectively be the genome algebra and the class algebra of $G$ with $Z$ and $\zee$ the symmetry element of $\cA$ and $\cA_0$. Then 
\begin{enumerate}
\item[(i)] $\zee$ is idempotent, $\zee$ is a left identity in $\cA$, and $\zee$ is the identity in $\cA_0$;
\item[(ii)] $\cA$ has a basis of the form $\{\zee g : g\in G\} = \{ \zee g_1, \ldots , \zee g_K\}$ and 
\[ K := \dim(\cA) = \big|\{[g] : g\in G\}\big|= \tfrac{|G|}{|Z|}\,;\]
\item[(iii)] $\cA_0$ has a basis of the form $\{\zee g \zee: g\in G\} = \{ \zee g_1\zee, \ldots , \zee g_L \zee\}$ and 
\[ L := \dim(\cA_0) = \big|\{[g]_D : g\in G\}\big|\,.\] \qed
\end{enumerate}
\end{proposition}

For the remainder of this section, we fix a basis for each of $\cA$ and $\cA_0$, as defined in (ii) and (iii) above.  

\begin{remark}
Each equivalence class $[g]_D$ can be viewed as an orbit of $G$ under an action of the group $Z\times Z$ and thus, from (iii) above, the dimension $L$ of $\cA_0$ may be calculated via Burnside's lemma \cite[Prop. 29.4]{James_Liebeck}. By combining this with the dual orthogonality relations on the group $G$, one can directly obtain the dimension result stated below in Theorem~\ref{thm:gen-chars-etc} (i). 
\end{remark}

We note that working `entirely' in the genome algebra does not mean that we forget about the group $G$. In practice, one would observe a genome with a particular orientation and reference frame, thus as an instance $g\in G$, and then identify the genome as the cloud $\zee g\in\cA$ for the purposes of computation. There are $K=\tfrac{|G|}{|Z|}$ distinct genomes, corresponding to the distinct basis elements $\zee g$ of the genome algebra $\cA$.  Similarly, one would conceive a rearrangement initially as an instance $a\in G$ and then lift to $\zee a$ in the genome algebra.
Considering all orientations of a rearrangement instance on all orientations of a genome corresponds to a left action of the algebra $\cA$ on itself. Distinct elements of $\cA$ that correspond to the same element of $\cA_0$ act indistinguishably, since
\begin{equation}\label{eq:zee a zee} (\zee a) \cdot (\zee g) = (\zee a\zee) \cdot (\zee g)\,.\end{equation}
Thus there are $L = \dim(\cA_0)$ distinct rearrangement actions.\footnote{We note that most of these mathematically possible rearrangements would not correspond to biologically plausible ones, so would not appear in rearrangement models in practice. \bt{For a deeper consideration of biologically plausible rearrangements from an algebraic perspective, see \cite{joshvjez}.}}

The (left) regular representations $\rho_{\reg}^{\cA}$ of the genome algebra $\cA$ and $\rho_{\reg}^{0}$ of the class algebra $\cA_0$ can be constructed in the usual way (c.f. (\ref{eq:regrepA})) via the bases fixed above. As in Section~\ref{sec:A}, one readily verifies that $\rho^{\cA}_{\reg}(\zee)$ is the $K\times K$ identity matrix and that $\rho^{\cA}_{\reg}(\zee g)^T = \rho^{\cA}_{\reg}(\zee g^{-1})$. Since $\zee$ is an identity in $\cA_0$, $\rho^0_{\reg}(\zee)$ is the $L\times L$ identity matrix.  In this case, the equivalence classes $[g]_D$ need not be the same size and thus, in general, $\rho^0_{\reg}(\zee g\zee)^T \neq \rho^0_{\reg}(\zee g^{-1}\zee)$.\footnote{One can verify via a simple counting argument that
$\big(\rho^0_{\reg}(\zee g\zee)\big)_{ij} 
		= \tfrac{\left|[g_i]_D\right|}{\left|[g_j]_D\right|}\big(\rho^0_{\reg}(\zee g^{-1}\zee)\big)_{ji}$.}
We denote the regular characters of $\cA$ and $\cA_0$ by $\chi_{\reg}^{\cA}$ and $\chi_{\reg}^0$ respectively and note that these take real values on any algebra element that is a real linear combination of basis elements.

Recall that, by Maschke's theorem \cite[Thm. 4.1.1]{etingof_book}, the group algebra $\CC[G]$ of any finite group $G$ can be written as a direct sum over its irreducible modules.

\begin{theorem}\label{thm:gen-chars-etc}
Let $G$ be a finite group with subgroup $Z\subseteq G$. Denote the distinct irreducible submodules of $\CC[G]$ by $V_i$, with $\dim(V_i)=D_i$ for each $i$ so that
\[ \CC[G] \cong \bigoplus_{i=1}^M D_i V_i\,, \]
and denote the corresponding irreducible representations and characters of $\CC[G]$ by $\rho_i$ and $\chi_i$ respectively. 
Then the following hold. 
\begin{enumerate}
\item[(i)] For each $i$, $W_i := \zee\cdot V_i$ is either $\{0\}$ or an irreducible  $\cA_0$-module and 
\begin{equation}\label{eq:decomp-gen-A0} \cA_0 \cong \bigoplus_{\substack{1\leq i\leq M\\ W_i\neq\{0\}}} k_i W_i \,,\end{equation}
with $k_i = \dim(W_i) = \chi_i(\zee)$ for each $i$. Thus, $\dim(\cA_0) = L = \ds\sum_{i=1}^M \chi_i(\zee)^2$.
\item[(ii)]  The modules $\{W_i : 1\leq i\leq M, W_i\neq\{0\} \}$ comprise all irreducible modules of $\cA$. Denoting the corresponding irreducible representations of $\cA$ and $\cA_0$ by $\rho_i^{\cA}$ and $\rho_i^{0}$ respectively,  
\[ \rho_i^{\cA}(\zee g) = \rho_i^{\cA}(\zee g\zee)  = \rho_i^{0}(\zee g\zee)\,,\]
for all  $g\in G$ and all (relevant) $1\leq i\leq M$. Thus for all $\bm{g}\in\cA_0$, $ \rho_i^{\cA}(\bm{g}) = \rho_i^{0}(\bm{g})$.
\item[(iii)] Denoting the corresponding characters of $\cA$ and $\cA_0$ by $\chi_i^{\cA}$ and $\chi_i^{0}$ respectively, and defining $\chi_i^{\cA} = \chi_i^0 \equiv 0$ for each $i$ such that $W_i=\{0\}$,
\[  \chi_i(\zee g) = \chi_i^{\cA}(\zee g)  = \chi_i^{0}(\zee g\zee)\,,\]
for all  $g\in G$ and all $1\leq i\leq M$.  Thus the characters $\chi_i, \chi_i^{\cA}$ and $\chi_i^0$ coincide on $\cA_0$.
\item[(iv)] For all $\bm{g}\in\cA$,
\[ \chi_{\reg}^{\cA}(\bm{g}) = \sum_{i=1}^M D_i \chi_i^{\cA}(\bm{g}) = \chi_{\reg}(\bm{g})\,,\]
where $\chi_{\reg}$ and $\chi_{\reg}^{\cA}$ denote the regular characters of $\CC[G]$ and $\cA$ respectively.
\item[(v)]  For any $\bm{g}\in\cA$, $\left(\tfrac{1}{K}\!\cdot \!\chi^{\cA}_{\reg}(\bm{g})\right)$ is the coefficient of $\zee$ in $\bm{g}$.
\end{enumerate}
\end{theorem} 
\begin{proof}
For (i), we use \cite[Prop. 4.18, Thm. 4.23]{steinberg_book}. For the remaining results, we use the observation (\ref{eq:zee a zee}) and proceed just as for the corresponding results in Section \ref{sec:A}. Note that in this general setting we cannot assume that the irreducible representations $\rho_i$ are orthogonal on $G$, but we can choose them to be unitary \cite[Thm. 4.6.2]{etingof_book}. This means that the corresponding irreducible representations of $\zee$ in $\CC[G]$ are self adjoint, and thus each $\rho_i(\zee)$ is unitarily diagonalisable, so that its eigenvectors form an orthonormal basis for $\CC^{D_i}\cong V_i$. Since the eigenvectors need not be real, the only difference in the proofs is that we need the conjugate transposes, not just transposes, of these vectors.
\end{proof}

The above results imply that $\cA_0 \cong \cA/\Rad(\cA)$ \cite[Thm. 3.5.4]{etingof_book}, which formalises the relationship between the genome algebra and the class algebra: $\cA_0$ is obtained from $\cA$ by factoring out the elements of $\cA$ that act trivially. We have previously expressed this as $\cA_0$ combining together the elements of $\cA$ that act indistinguishably. Another aspect of this is the following.

\begin{corollary}\label{coro:reps-same-gen}
Let $G$ be a finite group with subgroup $Z\subseteq G$. For any $g\in G$ and each irreducible representation $\rho^{\cA}_i$ of $\cA$, $\rho_i^{\cA}(\zee g) = \rho_i^{\cA}(\zee g^{\prime})$ for all $g^{\prime}\in[g]_D$. For any $g,g'\in G$, $\rho_{\reg}^{\cA}(\zee g) = \rho_{\reg}^{\cA}(\zee g^{\prime})$ if and only if $g^{\prime}\in[g]_D$. \qed
\end{corollary}

\begin{remark}\label{rem:semigroups}
We note that \cite[Prop. 7.14]{steinberg_book} implies part of Theorem~\ref{thm:gen-chars-etc} (iii) (namely, that $\chi_i^0(\bm{g})=\chi_i(\bm{g})$ for $\bm{g}\in\cA_0$) in the more general case of $G$ being a semigroup. Extending the framework to algebras based on semigroups would allow us to model \bt{further types of rearrangements, such as} insertions and deletions \cite{andrew14}, and we intend to investigate this possibility in future work.  
\end{remark}

We are ready to proceed with the formulation of path probabilities and the likelihood function within the genome algebra $\cA$. Firstly, we formally define a biological model for evolution in the genome algebra to be $(\M,w,dist)$, where
\begin{equation}\label{eq:model-A}
 \M:= \{ \zee a_1\zee, \zee a_2\zee , \ldots , \zee a_q\zee \} \subseteq \cA \,,
\end{equation}
for some $a_1,\ldots, a_q \in G$, $w: \M\to (0,1)$ is the probability distribution on $\cM$, and $dist$ is the probability distribution of rearrangement events in time. Note that we have used the form $\zee a\zee$ rather than $\zee a$ to avoid duplicating rearrangements in the model (that is, elements $\zee a, \zee a' \in\cA$ that are distinct but have the same left action, c.f. Examples~\ref{eg:zsigz} and \ref{eg:models}). Presently, we shall also add the condition that the model be {\em reversible}, that is, that $\zee a^{-1}\zee\in\cM$  for every $\zee a\zee\in\cM$, and $w(\zee a^{-1}\zee) = w(\zee a\zee)$.

We fix the reference genome to be $\zee\in\cA$ (whose instances are the elements of $Z$, in particular $e\in Z$). Then, for any target genome $\zee g$ (where we have observed the instance $g\in G$) and each $k\in\NN_0$, we define the path probability $\alpha_k(\zee g)$ to be
\[ \alpha_k(\zee g):= P(\zee \mapsto \zee g \textrm{ via } k \textrm{ rearrangements})\,, \]
where $\zee\mapsto\zee g$ means ``genome $\zee$ is transformed into genome $\zee g$''.
As usual, to find the path probability for an arbitrary genome $\zee h$ to be transformed into target genome $\zee g$, we can simply translate to the reference; that is, this is exactly the path probability for $\zee$ to be rearranged into $\zee gh^{-1}$, which is $\alpha_k(\zee g h^{-1})$.

Given a model $\M$, we define the corresponding {\em model element} of $\cA$ to be 
\[ \Aess:= \ds\sum_{i=1}^q w(\zee a_i\zee) \zee a_i\,.\]
We write $\Aess$ to distinguish the model element here from the previous definition of $\ess$ in the group algebra (c.f. Definition~\ref{def:mod-sym-els}), and choose to sum over rearrangements of the form $\zee a_i$ rather than $\zee a_i \zee$ for simplicity (recall from (\ref{eq:zee a zee}), these have the same action, so either form may be used).

It remains to connect the path probabilities to the regular character of powers of the model element (c.f (\ref{eq:s^k}) in Section~\ref{sec:perms}).  Recall from Section~\ref{subsec:markov} that $\zee a \cdot \zee g$ gives a convex combination of genomes, that is, 
\[ \zee a \cdot \zee g = \sum_{i=1}^K \texttt{p}_i \zee g_i \,,\]
where $\{\zee g_1, \ldots , \zee g_K\}$ is our fixed basis for $\cA$ and each $\texttt{p}_i$ is the proportion of the expansion of $\zee a \zee g$ that is equal to $\zee g_i$ or, equivalently, the probability that the rearrangement $\zee a$ acting on the genome $\zee g$ will result in the genome $\zee g_i$. Thus  
\[\Aess \cdot \zee = \sum_{j=1}^q w(\zee a_j \zee) \zee a_j \zee = \sum_{j=1}^q w(\zee a_j \zee) \sum_{i=1}^K \texttt{p}_{j,i} \zee g_i 
						=  \sum_{i=1}^K \left(\sum_{j=1}^q w(\zee a_j \zee)\texttt{p}_{j,i}\right) \zee g_i\,,\]
where we have rearranged and collected terms in the final step so that, for each $i$, $\sum_{j=1}^q w(\zee a_j \zee)\texttt{p}_{j,i}$ is the total probability that the genome $\zee$ will be transformed into the genome $\zee g_i$ via some (single) rearrangement chosen from the model. Thus
\[  \Aess \cdot \zee =  \sum_{i=1}^K \alpha_1(\zee g_i) \zee g_i \]
and, by repeatedly applying $\Aess$, one sees that 
\begin{equation}\label{eq:sk-gen}  \Aess^k \zee =  \sum_{i=1}^K \alpha_k(\zee g_i) \zee g_i\,. \end{equation}
Now, for $g\in G$ an instance of the genome of interest, multiply (\ref{eq:sk-gen}) on the right by $g^{-1}$ to obtain
\[ \Aess^k \zee g^{-1} 	=  \sum_{i=1}^K \alpha_k(\zee g_i) \zee g_i g^{-1} \,.\]
Since $\zee g_i g^{-1} = \zee$ if and only if $\zee g_i =\zee g$, we see that $\alpha_k(\zee g)$ is the coefficient of $\zee$ in the expansion of 
$\Aess^k \zee g^{-1}$, and thus
\begin{equation}\label{eq:gen-alpha-is-char}
\alpha_k(\zee g)  = \tfrac{1}{K}\!\cdot\chi_{\reg}^{\cA}\left(\Aess^k \zee g^{-1}\right)
\end{equation}
by Theorem~\ref{thm:gen-chars-etc} (v).

Theorem \ref{thm:gen-chars-etc} allows us to decompose the regular character in (\ref{eq:gen-alpha-is-char}) into irreducible characters, however we also will need to diagonalise the irreducible representation matrices. 

\begin{lemma}\label{lem:gen-diag-ess}
Let $G$ be a finite group with subgroup $Z\subseteq G$. Let $(\cM,w,dist)$ be a biological model for evolution of genomes represented by elements $\zee g\in\cA$ of the genome algebra (where $g\in G$) and let $\Aess\in\cA$ be the corresponding model element. If the model is reversible, then the following hold.
\begin{enumerate}
\item[(i)] The irreducible representation matrices of the model element $\Aess$ in $\cA$ are diagonalisable.
\item[(ii)] The regular representation of $\Aess$ in $\cA$ is symmetric. 
\end{enumerate}
\end{lemma}
\begin{proof}
(i) Suppose that the model is reversible and let $1\leq i\leq q$. 
We have (c.f. (\ref{eq:QrepA})) 
\[ \rho_i^{\cA}(\Aess) = \rho_i^{\cA}(\Aess\zee) := \overline{Q}^T\!\!\rho_i(\Aess\zee)Q\,,\]
where $Q$ is the $D_i\times k_i$ matrix of orthonormal eigenvectors for $\rho_i(\zee)$. 
Since $G$ is a finite group, we may choose the irreducible representation $\rho_i$ on $G$ to be unitary. Then, writing $ \rho_i(\Aess\zee)$ as a sum of matrices of the form $w(\zee a \zee)\left(\rho_i(\zee a\zee) + \rho_i(\zee a^{-1}\zee)\right)$ (omitting the second term if $a=a^{-1}$), each of which is self adjoint, we see that $\rho_i(\Aess\zee)$ is self-adjoint and thus so is $\rho_i^{\cA}(\Aess)$. For claim (ii), we proceed similarly, using the observation that $\rho_{\reg}^{\cA}(\zee a)^T = \rho_{\reg}^{\cA}(\zee a^{-1})$.
\end{proof}

\begin{theorem}\label{thm:likes-gen}
Let $G$ be a finite group with subgroup $Z\subseteq G$. Let $(\cM,w,dist)$ be a reversible biological model for evolution of genomes represented by basis elements of the genome algebra $\cA$ of $G$ with $Z$ and let $\Aess\in\cA$ be the corresponding model element. Let $g\in G$ be an observed instance of a genome $\zee g\in\cA$. Then the following hold.
\begin{enumerate}
\item[(i)] For any $k\in\NN_0$, the probability that the reference genome $\zee$ is transformed into the genome $\zee g$ via $k$ rearrangements chosen from the model is
\[ \alpha_k(\zee g)  = \tfrac{|Z|}{|G|} \chi_{\reg}^{\cA}(\Aess^k \zee g^{-1}) 
				=  \tfrac{|Z|}{|G|}\sum_{i=1}^M D_i \sum_{j=1}^{R_i} \lambda_{i,j}^k \mathrm{tr}(\rho_i^{\cA}(\zee g^{-1}) E^{\cA}_{i,j})\,,\]
where for each $i$, $E^{\cA}_{i,j}$ is the projection onto the eigenspace of the $j$th eigenvalue $\lambda_{i,j}$ of $\rho_i^{\cA}(\Aess)$.
\item[(ii)] If the distribution of rearrangement events in time is $dist=Poisson(1)$, then the probability that the reference genome $\zee$ is transformed into the genome $\zee g$ via the given model in time $T$ is given by the likelihood function
\[ L(T|g) =  e^{-T} \tfrac{|Z|}{|G|}\sum_{i=1}^M D_i \sum_{j=1}^{R_i}\tr(\rho^{\cA}_p(\zee g^{-1}) E^{\cA}_{i,j}) e^{\lambda_{i,j}T}\,.\]
\item[(iii)] For any genome $\zee h\in\cA$ with an instance $h\in[g]_D\cup[g^{-1}]_D$, the path probabilities and likelihood functions of $\zee g$ and $\zee h$ coincide. 
\end{enumerate}
\end{theorem}
\begin{proof}
The first expression for the path probability $\alpha_k(\zee g)$ was gained above (\ref{eq:gen-alpha-is-char}). To gain the second, we use the decomposition of the regular character from Theorem \ref{thm:gen-chars-etc} (iv), and then for each $i$, use the cyclicity of the trace to write
\[ \chi_i^{\cA}(\Aess^k \zee g^{-1})  
					= \tr\left(\rho_i^{\cA}(\zee g^{-1})\left(\rho_i^{\cA}(\Aess)\right)^k\right)  \,.\]
Then, from Lemma \ref{lem:gen-diag-ess}, we may diagonalise $\rho_i^{\cA}(\Aess)$ to gain the second expression. 
Analogously to the definition in Section~\ref{sec:perms}, but with genomes instead of elements of $G$, we define the likelihood function as
\[ L(T|g):= P(\zee g | T) =\sum_{k=0}^{\infty} \Prob(\zee\!\mapsto\!\zee g \textrm{ via } k \textrm{ events})
     														\Prob(k \textrm{ events in time } T)
				= \sum_{k=0}^{\infty} \alpha_k(\zee g) \tfrac{e^{-T}T^k}{k!}\,.\]
Then substituting in the expression from (i) and simplifying the power series gives (ii). \\
(iii) Let $h\in G$ such that $\zee h \zee = \zee g\zee$ or  $\zee h \zee = \zee g^{-1}\zee$. We show that $\alpha_k(\zee g) = \alpha_k(\zee h)$ for all $k\in\NN_0$, which implies (iii). Let $k\in\NN_0$. Since the trace is cyclic, we have
\[  \chi_{\reg}^{\cA}(\Aess^k \zee g^{-1}) = \tr\left(\rho_{\reg}^{\cA}(\zee g^{-1})\rho_{\reg}^{\cA}(\Aess)^k \right) =  \tr\left(\rho_{\reg}^{\cA}(\Aess)^k \rho_{\reg}^{\cA}(\zee g)\right) 
														=  \chi_{\reg}^{\cA}(\Aess^k \zee g) \,,\]
where the second equality was obtained by taking the transpose of the argument and applying  Lemma~\ref{lem:gen-diag-ess}.
Then from (i) it is clear that $\alpha_k(\zee g) =\alpha_k(\zee g^{-1})$ and these coincide with $\alpha_k(\zee h)$ by Corollary~\ref{coro:reps-same-gen}.
\end{proof}

Since the regular representation of the model element is symmetric, by Lemma \ref{lem:gen-diag-ess}, and the equilibrium distribution is the uniform distribution on the set of genomes, reversibility of the model $\M$ is equivalent to time reversibility of the underlying Markov process. As in Section~\ref{subsec:markov}, the regular representation of the model element in $\cA$ is the transition matrix for a Markov chain with states being genomes, with the probability that genome $\zee g_j$ transitions into genome $\zee g_i$ via $k$ rearrangement steps from the model given by
\[ \rho_{\reg}^{\cA}(\Aess^k)_{ij} = \alpha_k(\zee g_i g_j^{-1})\,.\]
Reversibility then means that for any genomes $\zee g,\zee h\in\cA$, the probability of $\zee g$ transforming into $\zee h$ in $k$ steps via the given model is the same as that of $\zee h$ transforming into $\zee g$ in $k$ steps. In terms of path probabilities, 
\[ \alpha_k(\zee gh^{-1}) = \alpha_k(\zee hg^{-1})\,,\]
which is just a special case of Theorem \ref{thm:likes-gen} (iii). Model reversibility thus implies that the MLE distance is `directionless', or symmetric, as is any evolutionary distance measure based on path probabilities calculated in this framework.

To conclude this section, we return briefly to the class algebra $\cA_0$. By constructing simple examples, one can verify that, in general, the regular representation matrices of non-identity basis elements in $\cA_0$ have non-zero entries on the diagonal, and thus see that we do not have an analogue of Theorem~\ref{thm:gen-chars-etc} (v) for $\cA_0$. That is, the regular character of $\cA_0$ is not counting occurrences of the identity in elements of this algebra, and thus cannot be used to calculate path probabilities as in Theorem~\ref{thm:likes-gen}.

Consider the underlying Markov model here, with transition matrix given by the regular representation of the model element analogue $\rho_{\reg}^0(\Aess\zee)$. 
Now the states are the basis elements $\zee g_i\zee$, each corresponding to an equivalence class $[g_i]_{D}$. 
Since each equivalence class $[g]_D$ is the disjoint union of $\tfrac{|[g]_D|}{|Z|}$ equivalence classes of the form $[g z]$ for some $z\in Z$, each basis element $\zee g\zee$ is the average of $\tfrac{|[g]_D|}{|Z|}$ distinct genomes of the form $\zee g z$  for some $z\in Z$. Thus an arbitrary element of the matrix gives us the {\em average} probability of a genome from a certain class transitioning into a genome from another class (and a diagonal element gives the probability of a transition within a class). This is not refined enough for our purposes, since, given one genome $\zee g$ and two more genomes $\zee g^{\prime}$ and $\zee g^{\prime\prime}$ that are in the same class ($g^{\prime} \in[g^{\prime\prime}]_D$), the probability of transitioning between $\zee g$ and $\zee g^{\prime}$ need not be the same as the probability of transitioning between $\zee g$ and $\zee g^{\prime\prime}$. 

The information is not entirely lost, however; one can use the first column of this Markov matrix to calculate path probabilities. Given an observed instance $g\in G$ of a genome, we find the appropriate basis element $\zee g_{\ell}\zee$ of $\cA_0$ such that $g\in [g_\ell]_D$, and then
\[ \alpha_k(\zee g) = \tfrac{|Z|}{|[g_{\ell}]_D|}\left(\rho_{\reg}^{0}((\Aess\zee)^k)\right)_{\ell 1}\,.\]

Of course, the fact that this path probability information exists in the regular representation does not mean that it is {\em easy} to obtain, in particular since the size of the regular representation in $\cA_0$ is likely to be rapidly increasing with the number of genomic regions (for example for $G=\SN$ and $Z=\DN$, $\dim(\cA_0)$ is proportional to $(N-2)!$) and, being unable to retain the `first column' information through diagonalisation (as one can for the trace), one would need to calculate the $k$th power of the matrix for each desired path probability. We further note that calculating the equivalence classes themselves, and checking for membership of an equivalence class, is a non-trivial exercise and simply not feasible for large numbers of regions.

In any case, the class algebra $\cA_0$ is nicer in some ways than the genome algebra $\cA$, in particular in that it is decomposable (that is, isomorphic to a direct sum of its irreducible modules). This property is formally known as {\em semisimplicity}. Then, since the irreducible modules of $\cA_0$ are identical to those of the algebra $\cA$ and the irreducible representations of the two algebras not only have the same dimension but coincide on the objects of interest,
\[ \rho^0_p(\zee g\zee) = \rho^{\cA}_p(\zee g\zee) = \rho^{\cA}_p(\zee g)\,,\]
one may in fact choose to implement the calculations of the irreducible representations in $\cA$ or $\cA_0$ and then, either way, combine the results together according to the decomposition given in Theorem~\ref{thm:likes-gen}.

\section{Conclusion}\label{sec:conc}

We have presented a coherent algebraic framework for modelling \bt{some classes of} genomes and rearrangements in an algebra that incorporates the inherent physical symmetries into each element. Algebraic frameworks for modelling genome rearrangement have been studied previously \cite{mei-dias,moul-steel,andrew14}, and the importance of including genome symmetry in rearrangement distance calculations has been recognised \cite{attilaand,mles}, however our unified approach, \bt{incorporating symmetry into the position paradigm framework \cite{wildcrazy}}, is new. 

Beginning with the specific case of circular genomes \bt{modelled} with unoriented regions and dihedral symmetry, we explicitly constructed the genome algebra from the symmetric group algebra, and showed that the MLE computations can be performed entirely within this algebra. By identifying genomes and rearrangements with single elements -- permutation clouds -- in the genome algebra, we have advanced previous work that identified genomes with cosets of permutations (where each element of a given coset represents an instance of a genome in a fixed physical orientation) but used the permutations as the basis elements for computation \cite{attilaand,mles,jezandpet,vjez_circ1}. We have both explained and removed the redundancy that we identified  \cite{vjez_circ1} in the implementation of the calculations in the symmetric group algebra.

In \cite{vjez_circ1}, we also signalled a desire to extend our technique for calculating the MLE to other settings, for example to include oriented regions or genomes with non-dihedral symmetry. We have not recorded the results of any explicit computations here, however we have algebraically verified that the technique can indeed be extended to a much more general case. For genomes where a single physical orientation can be represented by elements of a group $G$, and their physical symmetries by the subgroup $Z\subseteq G$, we defined the genome algebra of $G$ with $Z$; here, as in the special case described above, genomes and rearrangements correspond to basis elements (clouds). We showed that the path probabilities and thus the MLE can be formulated and the computations performed entirely in this genome algebra. \bt{An application of the} framework to modelling signed circular genomes, using the hyperoctahedral group and \bt{two possible} symmetry groups, \bt{is presented in \cite{vjoshjez}, along with the results of some sample computations that illustrate how the framework may be applied to compare different models and distance measures.} 

Although the genome algebra has lower dimension than the group algebra (by a factor of $\tfrac{1}{2N}$ in the $\DN$ case and $\tfrac{1}{|Z|}$ in the general case), this does not significantly reduce the computational complexity of calculating the MLE. \bt{We have performed distance calculations, in reasonable time, for genomes with up to twelve unoriented regions (unpublished) and up to six oriented regions \cite{vjoshjez}. Work to extend our initial experimental calculations to implementation of the framework for larger numbers of regions is ongoing. In particular, we are exploring the use of simulations and intend to apply} numerical approximations to make \bt{distance calculations tractable for genomes with larger} numbers of regions. 

\bt{Whilst the framework does not specify a particular rearrangement model (and indeed, allows choice both in the type of rearrangements allowed and their relative probabilities of occurring), we cannot currently model insertions, deletions, or duplications, since the underlying group structure means we are restricted to rearrangements that do not alter the set of regions. This is a clear limitation of the current approach. To address this, we are currently working on extending the framework to a semigroup-based approach, with the aim of accommodating insertions and deletions. Furthermore, whilst one can apply different probabilities to different rearrangements (and, depending on the genome's symmetry, rearrangements at different genomic positions), the current approach does not incorporate intergenic regions or explicitly consider breakpoints. Whether a group- or semigroup-based genome algebra approach can be devised that incorporates these biological realities, and others such as multiple chromosomes,  is another question for future research.}

\bt{Finally, we note that} the applications of this algebraic framework are not limited to calculating MLEs. The likelihood function is built from path probabilities; since our fundamental results hold for these `building blocks', other rearrangement distance measures that are based on path probabilities may be calculated via the genome algebra. We have shown that, via the regular representation of the genome algebra, the general genome rearrangement model can be viewed as a discrete (or, with the addition of the stochastic component, continuous time) Markov chain, and thus represented as a connected graph, generalising the Cayley graph approach \cite{moul-steel,chad-andrew}. \bt{This facilitates the calculation of further distance measures, for example mean first passage time, as demonstrated in \cite{vjoshjez}.}

\bibliographystyle{plain}
\bibliography{biblio}

\end{document}